\documentclass[onecolumn,a4paper,accepted=2026-07-13]{quantumarticle}

\pdfoutput=1

\usepackage{placeins}

\usepackage{amsmath}
\usepackage{amsfonts}
\usepackage{amssymb}
\usepackage{amsthm}
\usepackage{dsfont}
\usepackage{mathtools}
\usepackage{bbm}

\usepackage{physics}

\usepackage{marvosym}
\usepackage{wasysym}

\usepackage{graphicx}
\usepackage{enumerate}
\usepackage{hyperref}
\usepackage{float}
\usepackage[ruled,longend]{algorithm2e}
\usepackage{fancyref}
\usepackage{etoolbox}
\usepackage{appendix}
\usepackage{soul}
\usepackage[dvipsnames]{xcolor}
\usepackage[noabbrev,capitalise]{cleveref}
\usepackage[justification=raggedright,singlelinecheck=false]{caption}
\usepackage[justification=centering]{subcaption}

\usepackage[numbers,sort&compress]{natbib}

\makeatletter 
\patchcmd\caption@subtypehook{\let\label\subcaption@label}
{\let\label\subcaption@label\let\ltx@label\subcaption@label}{}{\fail}
\makeatother

\hypersetup{colorlinks,linkcolor=blue,citecolor=red,urlcolor=blue}

\creflabelformat{equation}{#2#1#3}

\theoremstyle{definition}
\newtheorem{definition}{Definition}
\newtheorem{proposition}{Proposition}

\newtheorem{corollary}{Corollary}

\newcommand{\errorchain}{e}
\newcommand{\errorclass}{\mathcal{C}}
\newcommand{\db}[1]{{{[\![ #1]\!]}}}

\newcommand{\edit}[1]{{#1}} 

\begin{document}

\title{A partition function framework for estimating logical error curves in stabilizer codes}
\author{Leon Wichette}
\email{leon.wichette@dfki.de}
\affiliation{%
    Robotics Innovation Center, German Research Center for Artificial Intelligence, 28359 Bremen,
    Germany
}
\author{Hans Hohenfeld}
\email{hans.hohenfeld@uni-bremen.de}
\affiliation{%
    Robotics Research Group, University of Bremen, 28359 Bremen, Germany
}
\author{Elie Mounzer}
\email{elie.mounzer@dfki.de}
\affiliation{%
    Robotics Innovation Center, German Research Center for Artificial Intelligence, 28359 Bremen,
    Germany
}
\author{Linnea Grans-Samuelsson}
\email{linnea.grans-samuelsson@physics.ox.ac.uk}
\affiliation{Rudolf Peierls Centre for Theoretical Physics, University of Oxford, Oxford OX1 3PU, United Kingdom}

\begin{abstract}
Based on the mapping between stabilizer quantum error correcting codes and disordered statistical
mechanics models, we define a ratio of partition functions that measures the success probability for
maximum partition function decoding, which at the Nishimori temperature corresponds to maximum
likelihood (ML) decoding. We show that this ratio differs from the similarly defined order
probability and describe the decoding strategy whose success rate is described by the order
probability. We refer to the latter as a probabilistic partition function decoding and show that it
is the strategy that at zero temperature corresponds to maximum probability (MP) decoding.
Based on the difference between the two decoders, we discuss the possibility of a maximum partition
function decodability boundary outside the order-disorder phase boundary. At zero temperature, the
difference between the two ratios measures to what degree MP decoding can be improved by accounting
for degeneracy among maximum probability errors, through methods such as ensembling.
We consider in detail the example of the toric code under bitflip noise, which maps to the Random
Bond Ising Model. We demonstrate that estimation of logical performance through decoding probability
and order probability is more sample efficient than estimation by counting failures of the
corresponding decoders\edit{, especially in the regime of low noise}.
We consider both uniform noise and noise where qubits are given individual error rates. The latter
noise model lifts the degeneracy among maximum probability errors, but we show that ensembling
remains useful as long as it also samples less probable errors. \edit{We also consider, in less detail, the color code under bitflip and depolarizing noise.}
\end{abstract}

\maketitle

\section{Introduction}

For large scale quantum computing, high-fidelity qubits and operations are needed. On hardware with
limited fidelities, this can be achieved through quantum error correction and fault-tolerant logical
operations on the encoded logical qubits. The logical performance can be improved by hardware
improvements that reduce the physical error rates, or by the usage of a better quantum error
correcting code or a more accurate decoder. In the choice of quantum error correcting code, logical
performance must be balanced against the spacetime overhead, while the decoder accuracy must be
balanced against the decoder runtime.

Just as there is currently a multitude of approaches to hardware design, the same holds for quantum
error correcting codes and decoding approaches, with the list likely to keep growing rapidly. Within
this abundance of options, estimates of the logical performance under hardware-realistic noise
models provide useful guidance. One important metric of code performance, the optimal threshold, has
for scalable families of stabilizer and subsystem codes early on been related to order-disorder
phase transitions in Random Bond Ising-type statistical mechanics models along a certain line in the
phase diagram: the \emph{Nishimori line}~\cite{dennis2002b}. This statistical mechanics mapping has
been extended from independent qubit Pauli noise to more realistic noise models, including circuit
noise and coherent noise~\cite{chubb2021c, rispler2024, piveteau2024, behrends2024d, behrends2025,
vodola2022, venn_coherent-error_2023}, and used to extract finite size corrections to the threshold for surface codes tailored
to biased noise~\cite{xiao2024a}. Similarly, the phase transition along the zero temperature line
has been shown to yield the threshold under maximum probability (MP)
decoding~\cite{dennis2002b,chubb2021c}. More specifically, maximum likelihood (ML) decoding relates
to the comparison of certain partition functions at the Nishimori temperature (as detailed in
Section~\ref{Background}), while maximum probability decoding relates to the comparison of partition
functions at zero temperature. In the toric code under independent identically distributed bitflip
noise, which maps to the $\pm J$ Random Bond Ising Model (RBIM), the phase boundary has been shown
to be reentrant~\cite{wang2003}, illustrating the non-optimality of MP decoding. More recently, the
RBIM phase boundary has been mapped out through numerical estimates of the order
probability~\cite{thomas2011b}.

While the order probability allows for the determination of the phase boundary, and hence the
threshold of an optimal decoder, we show that it does not measure the success rate of the optimal
decoder. Instead, it measures the success rate of what we shall refer to as a probabilistic
partition function decoder. The full optimal logical error curves are thus inaccessible through this
measure. Analogous to order probability, we introduce a ratio that we name the \emph{decoding
probability}, which on the Nishimori line measures the optimal logical error curves. We expect that
estimating optimal performance by measuring the decoding probability directly generally requires
less sampling overhead than estimation by counting the number of failures of the corresponding ML
decoder. Previous works concerning performance estimates via the statistical mechanics mapping have
mainly focused on establishing the optimal thresholds of error correcting codes under different
noise models (e.g.,~\cite{chubb2021c, rispler2024, bombin2012a, vodola2022, andrist2012a,song2022,
katzgraber2010, Kovalev:2014vhn}). 
In~\cite{bravyi2014b}, full optimal error curves were computed for the surface
code under independent identically distributed bitflip noise by counting the number of failures of
an ML decoder.
Tensor network methods have been used to approximate optimal decoding in a controllable
way~\cite{bravyi2014b,chubb2021c,piveteau2024}. Most recently,~\cite{piveteau2024} generalized
tensor network decoding to the case of $3$D codes, such as the $3$D unrotated surface code with
depolarizing noise, as well as $2$D codes with circuit level noise.
There has also been previous work on estimating optimal performance in the surface code through the
construction of a lookup table, in order to test the accuracy of different suboptimal surface code
decoders~\cite{maan2023a}, although the authors estimated that this method could at most be applied
to a distance $9$ surface code.
On the Nishimori line, measurement of optimal performance by decoding probability corresponds to the
measurement by syndrome-averaged minimum infidelity used in~\cite{behrends2025} (see also the
related work~\cite{behrends2024d} by the same authors) as well as the success rate shown
in~\cite{bravyi2014b}. In~\cite{thomas2011b}, the order probability in the Random Bond Ising Model
was linked to logical performance rates of the toric code, but the distinction between probabilistic
partition function decoding and optimal decoding was not made, and the resulting curves do not show
the success rate of an optimal decoder.  \edit{In more recent work, Ref.~\cite{english2025isingdonutregimestopological} focuses on the toric code under bitflip noise with full post-selection, which maps to a clean Ising model where analytic expressions can be derived.}

Generalizing beyond ML and MP decoders, one can define decoders that return a correction based on
either the maximization of the partition function at a given temperature, or that probabilistically
return a correction with the probabilities set by the partition functions. We refer to these as
maximum partition function decoders and probabilistic partition function decoders, respectively.
(For the former, previous works similarly define minimum free energy decoders~\cite{chubb2021c}.
However, in the zero temperature limit, we find that the partition function is the better quantity
to consider, since degeneracy can play an important role.) We show that the decoding probability
measures the success rates of the former, and the order probability measures the success rates of
the latter. This gives access to the logical error curves for any member of these two families of decoders. 
We relate probabilistic partition function decoding at zero temperature to MP decoding,
and maximum partition function decoding at zero temperature to ``degeneracy enhanced'' MP (dMP)
decoding -- meaning that, whenever there are multiple maximum probability errors, the decoder
returns a correction from the equivalence class containing the largest number of such errors. In
matching based decoders, degeneracy enhancement can be approximated through the method of
\emph{ensembling} described in~\cite{shutty2024}, where the edge weights of the decoding graph are
perturbed in order to sample over multiple matchings. Again, we expect that estimating the success
rates through order probability and decoding probability requires fewer samples than by counting the
number of failures of the corresponding decoders. Together, order probability and decoding
probability provide a flexible framework for estimating the performance under different decoding
schemes. Given a hardware realistic noise model and a set of potential codes that could run on the
hardware, such an estimation provides a way to filter out codes with low optimal performance before
spending resources on developing a fast decoder. In the development of a fast decoder tailored to a
specific code, it also provides a measure for how far from optimality the decoder is, and whether or
not there is enough room for improvement to motivate further refinements.

Other than being interesting as tools for performance estimation in stabilizer codes, we believe
that the distinction between decoding probability and order probability is of interest for further
characterizing the phase diagrams of the statistical mechanics models that the codes map to. The
shape of the RBIM phase boundary has historically attracted interest of its own, having originally
been expected to be vertical~\cite{nishimori1986} but having later been found to be
reentrant~\cite{wang2003, thomas2011b}. The distinction between order probability and decoding
probability opens the possibility that there could exist stabilizer codes whose associated
statistical mechanics models have maximum partition function decodability boundaries (defined from
the thresholds of maximum partition function decoders) that differ from their phase boundaries. We
define in Section~\ref{SuccessRates} conditions under which a model would possess a vertical
decodability boundary even with a reentrant phase boundary.

After describing the general framework for performance estimation via decoding probability and order
probability, we consider in detail the example of the toric code under both uniform (independent
identically distributed) and non-uniform bitflip noise, where in the latter noise model we assign
different error probabilities to each qubit. While still a simplified ``toy model'', non-uniform
qubit fidelities represent a scenario that better agrees with hardware observations~\cite{arute2019a}
and introduces a degeneracy lifting into the model where the distinction between
MP and dMP decoding vanishes, so that the zero temperature locations of the phase boundary and the
decodability boundary must agree. Meanwhile, we find that an improvement from ML over MP remains,
although it decreases for larger standard deviations. This implies that there is still room for
potential gains from decoder improvements through methods such as ensembling, although in the case
of ensembling, the perturbations must now be taken large enough to ensure sampling of not only the
most probable errors.

The structure of the paper is as follows.

In Section~\ref{Background} we summarize the relevant background. We describe ML, MP and dMP
decoding, and introduce the mapping between stabilizer codes and disordered statistical mechanics
models.

In Section~\ref{SuccessRates} we define maximum partition function decoders and probabilistic
partition function decoders, and show that their success rates are measured by the decoding
probability and the order probability, respectively. We discuss under which conditions a maximum
partition function decoder could reach optimality even away from the Nishimori line.

In Section~\ref{FKT_results_section} we present numerical results for the toric code under bitflip
noise, using the FKT algorithm to compute the relevant partition functions. We demonstrate that the
decoding probability and order probability give estimates of the success rates of ML, MP and dMP
decoders, \edit{and that significantly fewer samples are required to reach a given confidence interval
compared to estimation by counting the number of decoder failures. The relative required sample size depends on $p$ and $d$, and decreases in the low noise regime, making estimation through partition function ratios especially useful where regular sampling struggles.} 
Comparing the dMP and ML thresholds, we see that the maximum partition function decodability boundary is reentrant in the Random Bond Ising Model. 

In Section~\ref{BiasDegeneracyEnsemblingMatching} we focus on the zero temperature limit, and
compare to matching based decoders. We discuss the role of boundary conditions and the parity of the
code distance in determining the effect of degeneracy. We show how these factors influence the
amount of improvement observed from ensembling compared to regular MWPM in the toric code, unrotated
surface code and rotated surface code, using the \texttt{PyMatching}~\cite{higgott2025}
implementation of MWPM.\@ We also discuss bias, both in MWPM and in ensembling. \edit{In addition to regular bitflip noise, we briefly discuss the effects of degeneracy and bias for a noise model focusing on hook errors in the ``bare ancilla'' syndrome extraction circuit,
as well as for combined bitflip and readout noise.}

In Section~\ref{StrongEnsembling}, we turn to non-uniform bitflip noise in the toric code,
drawing qubit error rates from a uniform distribution with varying standard deviations. We estimate
MP, dMP and ML performance from the relevant partition functions, as well as the improvement from
ensembling over regular MWPM using \texttt{PyMatching}. We see that all decoding strategies perform
better under non-uniform noise than under uniform noise with the same mean error rate, and that MP
and dMP decoders now perform identically. Additionally, we see that the performance gains from
ensembling decrease as the standard deviation is increased.

\edit{Finally, in Section~\ref{ColorCode}, we briefly consider the color code under bitflip and depolarizing noise, focusing on the Nishimori line. In this setting, we use tensor network methods to estimate the partition functions. The results again indicate that the efficiency gain from estimating logical error rates via partition function ratios, instead of counting decoder failures, increases in the low noise regime, although limitations in reaching lower $p$ and larger $d$ make the observed gain less dramatic in the depolarizing case.}

\section{Background}%
\label{Background}

In this section we will present the preliminary material as well as the notation we use. We also
give a brief description of the role the statistical mechanics mapping plays in the estimation of
optimal logical performance. For a more detailed discussion, we refer the reader to~\cite{dennis2002b}.

\subsection{Preliminaries}%
\label{Preliminaries}

Consider the Hilbert space $\mathcal{H} = (\mathbb{C}^2)^{\otimes n}$ on $n$ qubits and the Pauli
group $\mathcal{P}$ with elements $g \in \mathcal{P}$ given by $g = \lambda g_1 \otimes g_2 \otimes
\cdots \otimes g_n$ for $\lambda \in \{ \pm 1, \pm i\}$ and where $g_i \in \mathcal{P}_i =
\{\mathbbm{1}, X, Y, Z\}$ are single Pauli operators. A stabilizer code is a quantum error
correcting code defined by a stabilizer group $\mathcal{S} \subset \mathcal{P}$, with $-\mathbbm{1}
\notin \mathcal{S}$, under which the code space $\mathcal{H}_{\vec{0}}$ is invariant, i.e.
\begin{equation}
    \mathcal{H}_{\vec{0}} = \{ \psi \in \mathcal{H} : S\psi = \psi, \,\; \forall S\in \mathcal{S}\}.
\end{equation}

For a stabilizer group of rank $r$, the code space encodes the quantum information of $n-r$ logical
qubits. The quantum information is subject to noise, described by a noise model. We here consider
noise models that contain only errors $e\in \mathcal{P}$. Each such error is assigned a probability
$p_e$. Coherent errors can be written as a decomposition in terms of elements in $\{\mathbbm{1}, X,
Y, Z\}$, and quantum error correction performed through repeated projective Pauli measurements leads
to a digitization of errors, so that the ability to correct a finite set of errors is enough to
correct any error~\cite{Nielsen_Chuang_2010}.

We consider a setting where the generators of the stabilizer group are repeatedly measured. If an
error is such that the system is in an eigenspace other than the code space after the following
round of stabilizer measurements, we refer to it as a \emph{detectable} error. From the stabilizer
measurements we obtain a string of $\pm 1$ eigenvalues $(-1)^{s_1}, (-1)^{s_2}, \ldots, (-1)^{s_r} $,
where $\vec{s} \in \mathbb{Z}_2^r$ is referred to as the \textit{syndrome}. The syndrome $\vec{s} =
\vec{0}$ is the trivial syndrome characterizing the code space.

In terms of the syndromes, the Hilbert space decomposes as:
\begin{equation}
    \mathcal{H} = \bigoplus_{\vec{s}} \mathcal{H}_{\vec{s}}.
\end{equation}
We say that a Pauli operator $f$ has syndrome $\vec{s}$ if and only if $f S_k = (-1)^{s_k}S_kf$ for
all \edit{stabilizer generators} $S_k \in \mathcal{S}$. The previous statement is also equivalent to saying that $f$
has syndrome $\vec{s}$ if and only if $f\mathcal{H}_0 = \mathcal{H}_{\vec{s}}$. Based on this
identification, we can speak not only of the syndrome of an eigenspace, but also of the syndrome of
an error.

The role of a \emph{decoder} is to, for any detectable error $e$, return a correction operator $f$
based on its measured syndrome $\vec{s}$, in such a way that $fe$ act trivially on the encoded
information. Let $\mathcal{C}(\mathcal{S})$ be the centralizer of the stabilizer group. For a given
syndrome $\vec{s}$, the set of Pauli operators with this syndrome is
$g(\vec{s})\mathcal{C}(\mathcal{S})$ for some fixed representative $g({\vec{s}})$. The centralizer
contains $\bar{L} \in \mathcal{C}(\mathcal{S})\backslash \mathcal{S}$, the logical Pauli operators,
as well as the stabilizers themselves. Since the stabilizer group is a subgroup of the centralizer,
the set of all Pauli operators with syndrome $\vec{s}$ can be partitioned into a disjoint union of
equivalence classes under stabilizer multiplication. Relative to a representative $g({\vec{s}})$, we
denote these classes as $\mathcal{C}_{\vec{s},\bar{L}} = g(\vec{s})\bar{L}\mathcal{S}$, referring to
them as logical equivalence classes. The decomposition is then given by
\begin{equation}
    g(\vec{s})\mathcal{C}(\mathcal{S}) = \bigcup_{\bar{L}} \mathcal{C}_{\vec{s},\bar{L}}.
\end{equation}
Taking the example of a stabilizer code with a single logical qubit, $ \bar{L} \in \{\mathbbm{1} ,
\bar{X}, \bar{Y}, \bar{Z}\}$ and the set of Pauli operators with syndrome $\vec{s}$ is partitioned
as follows:
\begin{equation}
    g(\vec{s})\mathcal{C}(\mathcal{S}) =
        \mathcal{C}_{\vec{s},\mathbbm{1}} \cup \mathcal{C}_{\vec{s},\bar{X}}
                                          \cup \mathcal{C}_{\vec{s},\bar{Y}}
                                          \cup \mathcal{C}_{\vec{s},\bar{Z}},
\end{equation}
with
\begin{alignat}{2}
    \mathcal{C}_{\vec{s}}^{(\mathbbm{1})}
        &= g({\vec{s}})\mathcal{S},
        &\quad\quad \mathcal{C}_{\vec{s},\bar{Y}}
        &= g({\vec{s}})\bar{Y}\mathcal{S}, \\
    \mathcal{C}_{\vec{s},\bar{X}}
        &= g({\vec{s}})\bar{X}\mathcal{S},
        &\quad\quad \mathcal{C}_{\vec{s},\bar{Z}}
        &= g({\vec{s}})\bar{Z}\mathcal{S}.
\end{alignat}

Based on these equivalence classes two Pauli operators $f$ and $g$ are considered to be equivalent
if they belong to the same logical equivalence class and we write $f \sim_\mathcal{S} g$. After
syndrome extraction, the task of a decoder boils down to finding a Pauli operator $g
\sim_\mathcal{S} e$ where $e$ is the error that occurred. The decoder succeeds if $g$ and $e$ belong
to the same logical equivalence class $\mathcal{C}_{\vec{s},\bar{L}}$, and fails otherwise.

The optimal decoding strategy is for the decoder to always return a recovery operation from the most
likely logical error class $\mathcal{C}_{\vec{s}}^{ML}$ for the given syndrome $\vec{s}$. After the
action of the noise channel, stabilizer measurement and extraction of syndrome $\vec{s}$, the
(unormalized) state can be written as the following linear combination:
\begin{equation}
    \rho(\vec{s}) =
        \sum_{\bar{L}} P(\mathcal{C}_{\vec{s},\bar{L}}| \vec{s})
        g(\vec{s})\bar{L} \rho \bar{L} g(\vec{s}),
\end{equation}
for some fixed, hermitian Pauli $g({\vec{s}})$, and logical operators $\bar{L}$. More precisely, a
\textit{maximum likelihood decoder} returns a Pauli operator $g \in \mathcal{C}_{\vec{s}}^{ML}$
where \edit{
\begin{equation}
    \mathcal{C}_{\vec{s}}^{ML} = \arg \max_{\mathcal{C}_{\vec{s},\bar{L}}}P(\mathcal{C}_{\vec{s},\bar{L}}| \vec{s}).
\end{equation}}
Maximum likelihood decoding succeeds with probability
\begin{equation}
    P^{ML}_{success} = \sum_{\vec{s}} P(\vec{s}) P(\mathcal{C}_{\vec{s}}^{ML}| \vec{s}).
\end{equation}

It is generally hard to compute the probabilities of the error classes. Meanwhile, it is for some
codes and noise models easy to find a maximum probability error. This motivates the alternative
strategy of \emph{maximum probability decoding}.  In this case, given a syndrome $\vec{s}$,  the
decoder outputs a recovery operation $e(\vec{s})$ such that $P(e(\vec{s})) \geq P(e'(\vec{s})) \quad
\forall e'(\vec{s})$.

In this case, the logical class $\mathcal{C}^{MP}_{\vec{s}}$ that the recovery operator belongs to
is given by \edit{
\begin{equation}
    \mathcal{C}^{MP}_{\vec{s}}
        = \arg \max_{\mathcal{C}_{\vec{s},\bar{L}}} \max_{\;\; e\in \mathcal{C}_{\vec{s}, {\bar{L}}}} P(e)
\end{equation}}
\edit{(chosen at random among all such classes in the case of degeneracy)} and the decoder succeeds with probability
\begin{equation}
    P^{MP}_{success} = \sum_{\vec{s}} P(\vec{s}) P(\mathcal{C}_{\vec{s}}^{MP}| \vec{s}).
\end{equation}

$P^{MP}_{success} \leq P^{ML}_{success}$, with strict inequality in general.

A third decoding strategy of interest in what follows is to return a Pauli operator belonging to the
class that contains the largest number $n_{\max}$ of maximum probability errors $e_{\max}$ (errors
such that $p(e_{\max}) \geq p(f) \;\forall f$). We refer to this as \textit{degeneracy enhanced MP
decoding} (dMP), as it also accounts for the degeneracy among maximum probability errors.
If all classes containing such errors have an equal number of maximum probability errors, dMP
decoding reduces to MP decoding. For instance, this holds whenever there is a unique maximum
probability error.

The logical class $\mathcal{C}^{dMP}_{\vec{s}}$ that the dMP recovery operator belongs to is given by \edit{
\begin{equation}
    \mathcal{C}^{dMP}_{\vec{s}} =
        \arg \max_{\mathcal{C}_{\vec{s},\bar{L}}}n_{\max}(\mathcal{C}_{\vec{s}, {\bar{L}}}| \vec{s})
\end{equation}}
and the decoder succeeds with probability
\begin{equation}
    P_{dMP}^{success} = \sum_{\vec{s}} P(\vec{s}) P(\mathcal{C}_{\vec{s}}^{dMP}| \vec{s}).
\end{equation}

In summary:

\begin{center}
\noindent\fbox{%
    \parbox{0.95\textwidth}{%
~

    {$\text{Correction returned by ML decoder: }
        \quad e({\vec{s}})  \in \mathcal{C}_{\vec{s}}
        \quad\text{ s.t. }
        \quad P(\mathcal{C}_{\vec{s}}) \geq P(\mathcal{C}'_{\vec{s}})
        \quad \forall \mathcal{C}'_{\vec{s}}$.}\\

    {$\text{Correction returned by MP decoder: }
        \quad e({\vec{s}})
        \quad \quad \quad \; \text{ s.t. }
        \quad P(e({\vec{s}})) \geq P(e'({\vec{s}}))
        \quad \forall e'({\vec{s}})$.} \\

    {$\text{Correction returned by dMP decoder:}
        \quad e({\vec{s}}) \in \mathcal{C}_{\vec{s}}
        \quad \text{ s.t. }
        \quad n_{\max}(\mathcal{C}_{\vec{s}}) \geq  n_{\max}(\mathcal{C}'_{\vec{s}} )
        \quad \forall \mathcal{C}'_{\vec{s}} $.}\\
    }}
\end{center}

\subsection{Error Class Probabilities in the Statistical Mechanical Mapping}%
\label{StatMechMapping}

Maximum likelihood decoding relies on computing error class probabilities.  For  stabilizer and
subsystem codes, this problem can be mapped to computing partition functions in disordered
statistical mechanics models\cite{dennis2002b}, after which existing tools for computation or
approximation of partitions functions can be used~\cite{thomas_exact_2009, wang_efficient_2001,
chubb2021e, piveteau2024}. In particular, this mapping relates the optimal threshold of quantum
error correcting codes to a critical point in an order-disorder phase transition in the statistical
mechanics model. The statistical mechanics mapping was initially done for the toric code under
bitflip noise, but has also been extended to other noise models such as depolarizing
noise~\cite{bombin2012a}, and correlated noise including circuit level noise~\cite{chubb2021c,
piveteau2024, rispler2024}. Additionally,~\cite{sriram2024} considers non-uniform noise models with
long-range spatio-temporal correlations for the repetition code and the toric code. In this section
we give an overview of the mapping, focusing on independent qubit noise in stabilizer codes and
referring the reader to~\cite{chubb2021c, pryadko2020} for details on the treatment of correlated
noise and subsystem codes. We present the example of the toric code under bitflip noise in more
detail.

To each \edit{stabilizer generator} $S_k$ we associate a classical spin degree of freedom $\sigma_k \in \{-1,1\}$.
The goal of the statistical mechanics mapping is to provide, for each error $e$, a Hamiltonian
$H_e(\{ \sigma_k \})$, such that the Boltzmann weight of the all spin up configuration
$\{\Uparrow\}$ gives the error probability, and the quenched disorder partition function $Z(e)$
(with $e$ seen as the disorder realization) gives error class probability:
\begin{align}
    \label{BoltzAndZ}
    e^{-\beta H_e(\{\Uparrow \})} &= P(e) \\
    Z(e)  = \sum_{ \{ \sigma_k \} } e^{-\beta H_e(\{ \sigma_k \})} & = P(\mathcal{C}(e))
\end{align}
for a suitable inverse temperature $\beta$. In what follows, it is often convenient to write the
partition function in terms of the density of states $g(E)$ as $Z(\errorclass) = \sum_E
g_{\errorclass}(E) e^{-\beta E}$, with $E$ the energy.

Below, we show the expression for such a Hamiltonian in a general stabilizer code, for a noise model
where each qubit $i$ fails independently, and where the qubit errors are single Pauli operators,
$g_i \in \mathcal{P}_i = \{ \mathbbm{1}, X,Y,Z \}$ with probabilities $p_i(\mathbbm{1}),
p_i(X),p_i(Y),p_i(Z)$.

Following~\cite{chubb2021c} we write the Hamiltonian in terms of the scalar commutator $\db{.\, ,.}:
\mathcal{P} \times \mathcal{P} \rightarrow \mathbb{C}$, which is defined by the following normalized
trace of the group commutator
\begin{equation}
    \db{A,B} \coloneq \frac{1}{2}\Tr[A,B],
\end{equation}
with the group commutator being given by $[A,B] \coloneq ABA^{-1}B^{-1}$.

The Hamiltonian takes the form \edit{
\begin{equation}
    \label{general_hamiltonian}
    H_e(\{\sigma_i\}) = -\sum_{i}\sum_{g_i \in \mathcal{P}_i}
        J_i(g_i) \, \db{g_i, e} \prod_{k:\db{g_i, S_k} = -1} \sigma_k
\end{equation}}
with coupling strength $J_i(g_i)$ defined by the Nishimori condition
\begin{equation}
    \beta J_i(g_i) = \frac{1}{|\mathcal{P}|}\sum_{f_i  \in \mathcal{P}_i}
        \log p_i(f_i)\db{g_i, f_i^{-1}}, \quad \forall i \in \{1,\ldots,n\}, g_i \in \mathcal{P}_i.
\end{equation}
The Hamiltonian in Eq.~\eqref{general_hamiltonian} is symmetric under stabilizer multiplication of
the error configuration as described in~\cite{chubb2021c}. This means that the Hamiltonian of an error $e'$,
which differs from another error $e$ by multiplication of a \edit{stabilizer generator} $S_k$, gives the same energy
for a given spin configuration as the Hamiltonian of the error $e$ gives for a related spin
configuration where the corresponding stabilizer spin degrees of freedom is flipped.

For the toric code under bitflip noise, this Hamiltonian reduces to the Hamiltonian of the Random
Bond Ising Model (RBIM). Here, the partition function under quenched disorder can be efficiently
computed by Pfaffian methods~\cite{thomas_numerically_2013}. The toric code and surface code under
bitflip noise will be the focus of the numerical examples presented in this paper. We show a sketch
of the RBIM phase diagram in Fig.~\ref{fig:RBIM-phase}.

\begin{figure}
    \centering
    \includegraphics[scale=1]{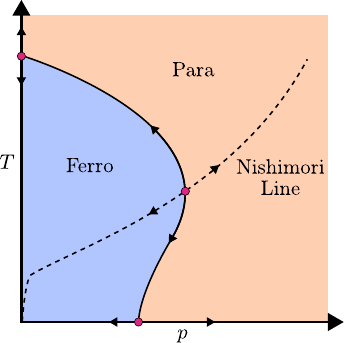}
    \caption{A sketch of the Random Bond Ising Model phase diagram, adapted from~\cite{picco2006},
    with the reentrance of the phase boundary exaggerated for the purpose of illustration. The
    arrows indicate the renormalization group flow.}
    \label{fig:RBIM-phase}
\end{figure}

For the toric code under independent identically distributed bitflip noise, the Hamiltonian of the
corresponding statistical mechanics model is defined, for each error configuration $\errorchain$, as
\begin{equation}\label{Hamiltonian_RBIM}
    H_\errorchain = - \sum_{\langle kl \rangle} J_{\langle kl \rangle} \sigma_k \sigma_l - K
\end{equation}
with the sum taken over all edges $\langle kl \rangle$. The spin degrees of freedom on the vertices
$k,l,\ldots$ correspond to the toric code stabilizers of $Z$-type, and the qubits in the toric code sit
on the edges $\langle kl\rangle$.

For an error $e$ the couplings are set to
\begin{equation}J_{\langle kl \rangle} =
    \begin{cases}
    -J,& \text{if } \langle kl \rangle \in \errorchain \\
    \phantom{-}J,              & \text{otherwise}
\end{cases}
\end{equation}
with the relation $e^{-2\beta J} = \frac{p}{1-p}$ and $e^{-2\beta K} = p(1-p)$ at the Nishimori
line. Here $\langle kl \rangle \in e$ means that the error $e$ has support on the qubit at edge
$\langle kl \rangle$.

In what follows, we keep the couplings $J$ fixed and vary the temperature. We denote by $T_{\text{Nish}}$
the temperature such that the above relation between $\beta_{\text{Nish}} =1/T_{\text{Nish}}$ and the couplings
$J,K$ is fulfilled. For identically distributed (uniform) bitflip noise, we fix $J=\pm 1$,
$T_{\text{Nish}}=\frac{2}{\ln((1-p)/p)}$. The Hamiltonian in Eq.~\eqref{Hamiltonian_RBIM} then fulfills
Eqs.~\eqref{BoltzAndZ}, so that a decoder that returns a correction based on the largest partition
function at $T=T_{\text{Nish}}$ is a maximum likelihood decoder.
The interaction strength can not uniformly be normalized in the case of non-uniform bitflip noise
with individually sampled qubit error probabilities $p_{i}$. Hence, for non-uniform bitflip noise,
we set the temperature fixed while varying the couplings:
$J_{i}=\frac{1}{2}\ln(\frac{1-p_{i}}{p_{i}})$ and $T_{\text{Nish}}=1$.

\section{Success rates for partition function decoders}%
\label{SuccessRates}

The goal of this section is to provide a unifying description of the success rates of ML, MP and dMP
decoders in terms of ratios of partition functions. Sampling these ratios generally allows for more
efficient estimation of the logical performance than directly sampling the number of decoding
failures.

\subsection{The maximum partition function decoder and the probabilistic partition function decoder}

The statistical mechanics mapping relates the probability of an error class $\mathcal{C}_{\vec{s}}$
to the partition function of any error $e\in \mathcal{C}_{\vec{s}}$ evaluated on the Nishimori line.
It is also of interest to consider the partition function at other temperatures. In particular, we
show how the zero temperature values measure the success rates of MP and dMP decoders. We note that
in the zero temperature analysis in~\cite{dennis2002b,chubb2021c} the focus has been on energies
rather than partition functions, which does not account for the contribution from degeneracy.

In what follows, we denote by $Z^T(\errorclass_{\vec{s}})$ the partition function of any
$\errorchain \in \mathcal{C}_{\vec{s}}$ evaluated at temperature $T$. For a given syndrome $\vec{s}$
and temperature $T$, we denote by $\{ \errorclass_{\vec{s}}^{\text{max}}(T)\} $ the set of all
classes $\errorclass_{\vec{s}}^{\text{max}}(T)$  such that
$Z^T(\errorclass_{\vec{s}}^{\text{max}}(T)) \geq Z^T(\errorclass_{\vec{s}}) $ for all
$\mathcal{C}_{\vec{s}}$.
Below the optimal threshold, in the infinite distance limit and at $T=T_{\text{Nish}}$, this set contains
only one element $ \mathcal{C}_{\vec{s}}^{\max}(T_{\text{Nish}})$, and its partition function increasingly
dominates over the others:
$Z^{T_{\text{Nish}}}(\errorclass_{\vec{s}}^{\max}(T_{\text{Nish}})) \to 1$ while
$Z^{T_{\text{Nish}}}(\errorclass_{\vec{s}})  \to 0$ for all $\errorclass_{\vec{s}} \neq
\errorclass_{\vec{s}}^{\max}(T_{\text{Nish}})$.\footnote{In the toric code, where the logical operators
correspond to nontrivial loops on the torus, an equivalent statement is that the free energy cost of
inserting a system-spanning domain wall diverges on the Nishimori line below the optimal threshold.}
At temperatures $T\neq T_{\text{Nish}}$, the set $\{\mathcal{C}_{\vec{s}}^{\max}(T)\}$ may contain more
than one element, even at noise levels below the optimal threshold. In the infinite temperature
limit, it will contain all classes consistent with a given syndrome, while in the zero temperature
limit it will contain all classes that contain the largest number $n_{\max}$ of maximum probability
errors.

For a given temperature $T$, we define the following two partition function based decoding
strategies:

\begin{definition}
    \label{maximum_Z_decoder_definition}
    A \textit{maximum partition function decoder at temperature T} is a decoder that, for each
    syndrome $\vec{s}$, returns a recovery operation belonging to one of the classes in $\{
    \errorclass_{\vec{s}}^{\text{max}}(T) \}$, chosen at random each time the decoder is called,
    with equal probability for each such class.
\end{definition}

\begin{definition}
    \label{probabilistic_Z_decoder_definition}
    A \textit{probabilistic partition function decoder at temperature T} is a decoder that, for each
    syndrome $\vec{s}$, returns a recovery operation belonging to a class
    $\tilde{\errorclass}_{\vec{s}}$ with probability
    \begin{equation}
        P_{\text{decoder}}(\tilde{\errorclass}_{\vec{s}})
            = \frac{Z^T(\tilde{\errorclass}_{\vec{s}})}{\sum_{\errorclass_{\vec{s}}}
            Z^T(\errorclass_{\vec{s}})}.
    \end{equation}
\end{definition}

In the following proposition, the first statement (which we include for completeness) has been well
established, going back to~\cite{dennis2002b}, while the distinction between the second and the
third statements has to the best of our knowledge not been made.

\begin{proposition}
    \label{decoder_statements}
    The following three statements hold:

    \textit{1.} A maximum partition function decoder at $T=T_{\text{Nish}}$ is an ML decoder.

    \textit{2.} A maximum partition function decoder at $T=0$ is a dMP decoder.

    \textit{3.} A probabilistic partition function decoder at $T=0$ is an MP decoder
\end{proposition}

\begin{proof}
    We prove each of the three statements in turn.

    \vspace{0.2cm}
    \textit{1.} The proof of this statement immediately follows from $P(\errorclass_{\vec{s}}) =
        Z^{T_{\text{Nish}}}(\errorclass_{\vec{s}})$.

    \vspace{0.2cm}
    \textit{2.} Writing the partition function as $Z^T(\errorclass_{\vec{s}}) = \sum_E
        g_{\errorclass_{\vec{s}}}(E) e^{-E/T}$, with $g_{\errorclass_{\vec{s}}}(E)$ the density of
        states for the class $\errorclass_{\vec{s}}$, we denote by $E_{\min}(\errorclass_{\vec{s}})
        $ the lowest energy $E$ such that $g_{\errorclass_{\vec{s}}}(E) \neq 0$, and denote by
        $E_{\min}(\vec{s}) = \min_{\errorclass_{\vec{s}}} E_{\min}(\errorclass_{\vec{s}}) $ the
        lowest energy among all error classes consistent with the syndrome $\vec{s}$.  Normalizing
        by the lowest energy Boltzmann weight,
        \begin{equation}
            \lim_{T
            \to 0} Z^T(\errorclass_{\vec{s}})e^{ E_{\min}(\vec{s})/T}
                = \begin{cases}
                    g(E_{\min}(\errorclass_{\vec{s}}))  &
                        \text{if } E_{\min}(\errorclass_{\vec{s}}) = E_{\min}(\vec{s}) \\
                    0 & \text{else}
            \end{cases}
        \end{equation}
        with $g(E_{\min}(\errorclass_{\vec{s}}))$ the number of lowest energy disorder realizations
        present in $\errorclass_{\vec{s}}$. Since the probability of an error is given by its
        Boltzmann weight at $T_{\text{Nish}}$, and the Boltzmann weight monotonically  decreases with $E$,
        the lowest energy errors are the most probable and
        $Z^T(\errorclass_{\vec{s}})e^{E_{\min}(\vec{s})/T} = n_{\max}(\errorclass_{\vec{s}})$, from
        which the second statement follows.

    \vspace{0.2cm}
    \textit{3.} Finally, a probabilistic partition function decoder at zero temperature returns a
        correction belonging to $\tilde{\errorclass}_{\vec{s}}$ with a probability
        \begin{equation}
            P_{\text{decoder}}(\tilde{\errorclass}_{\vec{s}})
                = \frac{n_{\max}(\tilde{\errorclass}_{\vec{s}})}{\sum_{\errorclass_{\vec{s}}}
                    n_{\max}(\errorclass_{\vec{s}})},
        \end{equation}
        which is equivalent to returning a maximum probability error $e$, with equal probability for
        each such error.
\end{proof}

\subsection{The decoding probability and the order probability}

We next introduce decoding probability, alongside the definition of order probability used
in~\cite{thomas2011b}, and show that the former measures the success rate of maximum partition
function decoders, while the latter measures the success rate of probabilistic partition function
decoders.

\begin{definition}
    \label{decoding_probability_definition}
    For a temperature $T$ and syndrome $\vec{s}$, let $\errorclass_{\vec{s}}^*(T)$ be an error class
    chosen at random from  $\{\errorclass^{\text{max}}_{\vec{s}}(T)\}$ with uniform probability. The
    \emph{decoding probability at temperature $T$} is defined as the following average over disorder
    realizations $e$ and class choices $\errorclass^*$:

    \begin{equation}\label{decoding_probability_definition_eq}
        P_{\text{d}}(T) = \left\langle
            \frac{%
                Z^{T_{\text{Nish}}}(\errorclass_{\vec{s}}^*(T))
            }{%
                \sum_{\errorclass_{\vec{s}}} Z^{T_{\text{Nish}}}(\errorclass_{\vec{s}})
            } \right\rangle_{\!\!e,\errorclass^*}
        \equiv
        \sum_e P(e) \sum_{\errorclass_{\vec{s}(e)}^*(T) }
            \frac{1}{|\{\errorclass^{\max}_{\vec{s}(e)}(T) \}|}
            \frac{Z^{T_{\text{Nish}}}(\errorclass_{\vec{s}(e)}^*(T))}{%
                \sum_{\errorclass_{\vec{s}(e)}} Z^{T_{\text{Nish}}}(\errorclass_{\vec{s}(e)})}.
    \end{equation}
\end{definition}

\begin{definition}
    \label{order_probability_definition}
    For each disorder realization $e$, let $\tilde{\errorclass}_{\vec{s}}$ be its error class, $e\in
    \tilde{\errorclass}_{\vec{s}}$.
    The \emph{order probability at temperature $T$} is defined as the following average over
    disorder realizations $e$:

    \begin{equation}
        P_{\text{o}}(T) = \left\langle
            \frac{Z^T(\tilde{\errorclass}_{\vec{s}})}{%
                \sum_{\errorclass_{\vec{s}}} Z^T(\errorclass_{\vec{s}})}
        \right\rangle_{\!\!e}
        \equiv
        \sum_e P(e)
        \frac{Z^T(\tilde{\errorclass}_{\vec{s}})}{\sum_{\errorclass_{\vec{s}}} Z^T(\errorclass_{\vec{s}})}.
    \end{equation}
\end{definition}

We note that the sum over partition functions in the denominator of
Eq.~\eqref{decoding_probability_definition_eq} is equal to one, as the partition functions at
$T_{\text{Nish}}$ are equal to class probabilities. Including the sum explicitly has the advantage of
making the expression valid also when the partition functions are computed only up to an overall
normalization.

While both the decoding probability and the order probability go to $1/N$ in the limit $T\to\infty$,
with $N$ being the number of equivalence classes for each syndrome, their values at finite
temperature are generally different. At  $T=0$ their values are equal if, for each syndrome
$\vec{s}$, \edit{$n_{\max}$ is identical for all classes in $\{ \errorclass_{\vec{s}}^{\max}\}$}.

\begin{proposition}
    \label{decoding_proposition}
    The success rate of a maximum partition function decoder at temperature $T$ is equal to the
    decoding probability at temperature $T$.
\end{proposition}

\begin{proof}
     We separate the averages over disorder realizations $e$ into averages for each syndrome $\vec{s}$,
    \begin{equation}
        P_{\text{success}} = \sum_{\vec{s}} P(\vec{s})P_{\text{success}}(\vec{s})
        =
        \sum_{\vec{s}} P(\vec{s})
        \sum_{\errorclass_{\vec{s}}}
            P(\errorclass_{\vec{s}}|\vec{s})
            P_{\text{success}}(\errorclass_{\vec{s}}).
    \end{equation}
    For disorder realizations $e\in C_{\vec{s}}$, the decoder succeeds if the recovery operator
    belongs to $\errorclass_{\vec{s}}$. By definition~\ref{maximum_Z_decoder_definition}, the
    recovery operator for a syndrome $\vec{s}$ is chosen at random among
    $\{\errorclass_{\vec{s}}^{\max}(T) \}$ with probability
    $P_{\text{decoder}}(\errorclass_{\vec{s}}^{*}(T)) =
    \frac{1}{|\{\errorclass_{\vec{s}}^{\max}(T)\}|}$. Hence,
    \begin{equation}
        P_{\text{success}}(\errorclass_{\vec{s}})
        = \begin{cases}
            \frac{1}{|\{\errorclass_{\vec{s}}^{\max}(T)\}|}
                & \text{if } \errorclass_{\vec{s}} \in \{  \errorclass_{\vec{s}}^{\max} (T)\}\\
            0 & \text{else},
        \end{cases}
    \end{equation}
    so that
    \begin{equation}
         P_{\text{success}}(\vec{s}) =
            \sum_{C^* \in \{\errorclass_{\vec{s}}^{\max}(T)\}}
                P(C^*|\vec{s}) \frac{1}{| \{\errorclass_{\vec{s}}^{\max}(T)\} |}.
    \end{equation}
    Finally we substitute
    \begin{align}
        P(\mathcal{C}^*|\vec{s})
            &= Z^{T_{\text{Nish}}} (C^*)
                = \frac{Z^{T_{\text{Nish}}}(C^*)}{ \sum_{C_{\vec{s}}} Z^{T_{\text{Nish}}}(C_{\vec{s}}) },\\
        P(\vec{s}) &= \sum_{\substack{e \\ \vec{s}(e) = \vec{s}}} P(e)
    \end{align}
    to get
    \begin{align}
         P_{\text{success}}
            &= \sum_{\vec{s}} P(\vec{s})
               \sum_{\errorclass^* \in \{\errorclass_{\vec{s}}^{\max}(T)\}}
               \frac{1}{| \{\errorclass_{\vec{s}}^{\max}(T)\} |}
               \frac{Z^{T_{\text{Nish}}}(\errorclass^*)}{%
               \sum_{\errorclass_{\vec{s}}} Z^{T_{\text{Nish}}}(\errorclass_{\vec{s}}) }\notag \\
            &= \sum_e P(e)
               \sum_{\errorclass^* \in \{\errorclass_{\vec{s}(e)}^{\max}(T)\} }
               \frac{1}{|\{ \errorclass^{\max}_{\vec{s}(e)} (T) \}|}
               \frac{Z^{T_{\text{Nish}}}(\errorclass^*)}{%
               \sum_{\errorclass_{\vec{s}(e)}} Z^{T_{\text{Nish}}}(\errorclass_{\vec{s}(e)})} .
    \end{align}
\end{proof}

We note that at the Nishimori temperature, in a setting with two classes per syndrome, the above
per-syndrome measure of logical performance reduces to the following measure given
in~\cite{behrends2025}:
\begin{equation}
    P_{\text{failure}} = \sum_{\vec{s}} \min_q P_{q,\vec{s}},
\end{equation}
where the authors denote by $P_{0,\vec{s}}, P_{1,\vec{s}}$ the class probabilities for the two
classes with syndrome $\vec{s}$.

\begin{proposition}
    \label{order_proposition}
    The success rate of a probabilistic partition function decoder at temperature $T$ is equal to
    the order probability at temperature $T$.
\end{proposition}

\begin{proof}
    The proof follows the same structure as the proof of Proposition~\ref{decoding_proposition}. The
    success rate can again be considered on a per-syndrome basis, with the decoder succeeding
    whenever the class encountered is the same as the class chosen by the decoder in accordance to
    Definition~\ref{probabilistic_Z_decoder_definition}:
    \begin{align}
        P_{\text{success}}
            &= \sum_{\vec{s}} P(\vec{s})
               \sum_{\tilde{\errorclass}_{\vec{s}}}
               P(\tilde{\errorclass}_{\vec{s}}|\vec{s})
               P_{\text{decoder}}(\tilde{\errorclass}_{\vec{s}})
             = \sum_{\vec{s}} P(\vec{s})
               \sum_{\tilde{\errorclass}_{\vec{s}}} P(\tilde{\errorclass}_{\vec{s}}|\vec{s})
               \frac{Z^T(\tilde{\errorclass}_{\vec{s}})}{%
                    \sum_{\errorclass_{\vec{s}}} Z^T(\errorclass_{\vec{s}})} \notag \\
            &= \sum_e P(e)
                \frac{Z^T(\tilde{\errorclass}_{\vec{s}}(e))}{%
                \sum_{\errorclass_{\vec{s}}} Z^T(\errorclass_{\vec{s}})}.
    \end{align}
\end{proof}

For $T=T_{\text{Nish}}$ the ratio averaged over to obtain the order probability, $\frac{Z^T(\tilde{\errorclass}_{\vec{s}})}{\sum_{\errorclass_{\vec{s}}} Z^T(\errorclass_{\vec{s}})}$, also appears within the expression for coherent information of CSS codes under depolarizing noise found in Ref.~\cite{colmenarez2024fundamental}. Coherent information provides an alternative method for finding the optimal threshold that can provide accurate estimates at low distances \cite{colmenarez_accurate_2024}, but it generally only provides bounds on the optimal logical success rate itself. An example of a lower bound from coherent information is shown for CSS codes under decoherence in Ref.~\cite{Niwa_2025}.  \edit{At $T=T_{\rm Nish}$, the probabilistic partition function decoder is also referred to as the \emph{sampling decoder} and, in qubit CSS codes under Pauli noise, is shown in Ref. \cite{kim2026optimalrecoveryquantumerror} to correspond to the Petz recovery channel and to share threshold with the optimal decoder.}

Not only do the ratios of Definitions~\ref{decoding_probability_definition}
and~\ref{order_probability_definition} measure the success rates of the decoders defined in
Definitions~\ref{maximum_Z_decoder_definition} and~\ref{probabilistic_Z_decoder_definition} when
averaged over all disorder realizations and class choices, but we expect that for any given finite
sample size, an estimate of the success rate computed from these ratios will generally have a
narrower confidence interval than an estimate based on counting the number of successes by the
corresponding decoder. For each disorder realization, the partition function ratios that are sampled
over to estimate decoding probability and order probability,
\begin{align}
    \mathcal{O}_{\text{d}}
        &= \frac{Z^{T_{\text{Nish}}}(\errorclass_{\vec{s}}^*(T))}{%
            \sum_{\errorclass_{\vec{s}}} Z^{T_{\text{Nish}}}(\errorclass_{\vec{s}})}  \\
    \mathcal{O}_{\text{o}}
        &= \frac{Z^T(\tilde{\errorclass}_{\vec{s}})}{%
            \sum_{\errorclass_{\vec{s}}} Z^T(\errorclass_{\vec{s}})},
\end{align}
will yield values in the interval $[0,1]$, which narrows to $[1/N,1]$ on the Nishimori line for
$\mathcal{O}_{\text{d}}$\footnote{For an error class $\errorclass \in
\{\errorclass^\text{max}_{\vec{s}}(T_{\text{Nish}})\}$, $Z^{T_{\text{Nish}}}(\errorclass) = P(\errorclass) \geq
\frac{1}{N}$.}. Meanwhile, counting the number of successes means sampling over
\begin{equation}
    \mathcal{O}^{\prime} = \begin{cases}
        1 & \text{ if success}\\
        0 & \text{ else}
    \end{cases}
\end{equation}
which only takes values in $\{0,1\}$, leading to a larger sample variance.

\subsection{Performance comparisons}

There are two natural comparisons for the two types of decoder strategies defined above:
performance differences for the same strategy as $T$ is varied, and  performance difference between
the two strategies for a fixed $T$.

It has been shown that when $T$ is varied from $T_{\text{Nish}}$ to zero in the toric code under
independent identically distributed bitflip noise, the phase boundary of the corresponding
statistical mechanics model is reentrant. This shows that the probabilistic partition function
decoder has a lower threshold at zero temperature than at $T_{\text{Nish}}$ in this setting. In models with
a reentrant phase boundary, it is interesting to consider whether or not the decodability boundary
-- defined by the threshold of the maximum partition function decoder -- is also reentrant. The
distinction between decoding probability and order probability opens the possibility that, as long
as the noise rate is below the optimal threshold, maximum partition function decoders can succeed
outside the phase boundary. (We expect that a maximum partition function decoder generally has
better performance than a probabilistic partition function decoder. At $T=T_{\text{Nish}}$ this is clearly
the case: the maximum partition function decoder is optimal, while the probabilistic partition
function decoder is suboptimal unless all error classes are equally probable.)

When comparing the present discussion to the phase boundary discussion in~\cite{chubb2021c}, it is
important to distinguish between a maximum partition function decoder and a minimum free energy
decoder. The authors of~\cite{chubb2021c} consider the latter. They note that the negative logarithm
defining the free energy $F$ from the partition function,
\begin{equation}
    F(T) = -T \ln Z(T),
\end{equation}
is monotonically decreasing, and relate the two decoding strategies to each other. However, at zero
temperature the free energy decoder will lose the information about degeneracy,
\begin{align}
    \lim_{T\to 0}F(T)
        &= \lim_{T\to 0}  -T \ln (\sum_E g(E) e^{-E/T})\notag \\
        &= \lim_{T\to 0} -T \ln(g(E_{\min})) - T\ln e^{-E_{\min}/T} -T
            \ln ( \sum_{E'>E_{\min}}g(E')e^{-E'/T})\notag \\
        &= E_{\min},
\end{align}
making it an MP decoder at $T=0$ even though it is a maximum partition function decoder for all $T>0$.

Making this distinction, it is clear that while neither the decodability boundary nor phase boundary
can extend further to the right than the optimal threshold, lemma 3 of~\cite{chubb2021c} should at
$T>0$ be seen as a statement about the decodability boundary rather than the phase boundary. To
phrase the possible distinction between the boundaries differently: even if the free energy
difference between error classes does not diverge at $T\neq T_{\text{Nish}}$, the minimum free energy
decoder may still succeed at this temperature, provided that the difference remains nonzero and the
free energy does diverge at $T_{\text{Nish}}$. Indeed, in such a case a minimum free energy decoder can
even retain optimality in spite of a reentrant phase boundary, apart from at $T=0$. To make a
statement that also holds at $T=0$, we consider a maximum partition function decoder instead: A
maximum partition function decoder can retain optimality at $T\neq T_{\text{Nish}}$, as long as at least
one of the most likely classes $\errorclass^* \in \{\errorclass^{\max}_{\vec{s}}(T_{\text{Nish}})\}$ has a
partition function that remains larger than that of any less likely class.\footnote{For noise rates
below the optimal threshold there is only one such class at large enough distance, but for finite
distances there may be more than one.} In such cases the decodability boundary is vertical.
\begin{corollary}
    \label{optimality_proposition}
    The success rate of a maximum partition function decoder at temperature $T$ is equal to the
    success rate of an optimal decoder if there exists at least one class
    $\errorclass^*_{\vec{s}} \in \{\errorclass^{\max}_{\vec{s}}(T_{\text{Nish}})\} $ such that
    $Z^T(\errorclass^*_{\vec{s}}) > Z^T(\errorclass'_{\vec{s}})$ for all $\errorclass'_{\vec{s}}
    \notin \{ \errorclass^{\max}_{\vec{s}}(T_{\text{Nish}}) \}$.
\end{corollary}

\begin{proof}
    By Definition~\ref{maximum_Z_decoder_definition}, if the criterion in
    Corollary~\ref{optimality_proposition} is fulfilled the maximum partition function decoder at
    temperature $T$ will return a maximally likely error class. By
    Definition~\ref{decoding_probability_definition} it will then have the same success rate as a
    maximum partition function decoder at $T_{\text{Nish}}$, which by Proposition~\ref{decoder_statements}
    is optimal.
\end{proof}

The above shows, in particular, that the thresholds of an MP decoder and a dMP decoder can differ
for codes and noise models where the phase boundary is reentrant. This brings us to the second type
of comparison: the performance difference between the two decoding strategies at a fixed
temperature. In the ordered phase at increasing distance, the two decoding strategies -- and hence
the two success rates -- will increasingly agree, as all partition functions but one go to zero.
However, the success rates may differ at finite distance even in the ordered phase, including in
cases where their thresholds are the same. In the next section, we study numerically how the
distinction between the two decoding strategies plays out in the toric code under bitflip noise,
focusing on $T=T_{\text{Nish}}$ and $T=0$.

We sum up the present section with the three ratios to be sampled over in order to estimate the
performance of the three decoders summarized at the end of Section~\ref{Preliminaries}:
\begin{center}
\noindent\fbox{%
    \parbox{0.99\textwidth}{%
\edit{
    \noindent For each disorder realization $e$, let $\tilde{\errorclass}_{\vec{s}(e)}$ be its error class and
    $\errorclass^*_{\vec{s}(e) }(T) \in \{ \errorclass^{\max}_{\vec{s} }(T) \}$ be chosen at random.}\\

    {$\text{Estimator for the performance of an ML decoder: }
        \frac{Z^{T_{\text{Nish}}}(\errorclass_{\vec{s}}^*(T))}{%
            \sum_{\errorclass_{\vec{s}}} Z^{T_{\text{Nish}}}(\errorclass_{\vec{s}})}
        \text{ at $T=T_{\text{Nish}}$.} $}\\

    {$\text{Estimator for the performance of an MP decoder: }
        \frac{Z^T(\tilde{\errorclass}_{\vec{s}})}{%
            \sum_{\errorclass_{\vec{s}}} Z^T(\errorclass_{\vec{s}})}
        \text{ at $T=0$.}$}\\

    {$\text{Estimator for the performance of a dMP decoder: }
        \frac{Z^{T_{\text{Nish}}}(\errorclass_{\vec{s}}^*(T))}{%
            \sum_{\errorclass_{\vec{s}}} Z^{T_{\text{Nish}}}(\errorclass_{\vec{s}})}
        \text{ at $T=0$.} $}\\
    }}
\end{center}

\section{Partition function based estimates of decoder performance in the toric code under bitflip noise}%
\label{FKT_results_section}

In this section we focus on the toric code under bitflip noise, mapped to the RBIM as detailed in
Section~\ref{StatMechMapping}. We numerically demonstrate Proposition~\ref{decoding_proposition} and
Proposition~\ref{order_proposition}, and also demonstrate that less samples are needed to reach a
certain precision (as measured by the width of the confidence interval) when estimating failure
rates using order probability and decoding probability than when counting decoder failures.

We mainly use Pfaffian methods to compute RBIM partition functions, adapting an implementation of
the Fisher–Kasteleyn–Temperley (FKT) algorithm by Thomas and
Middleton~\cite{thomas_numerically_2013} found
\href{https://github.com/a-alan-middleton/IsingPartitionFn}{here}. As a complement to this approach,
we also adapt a Wang-Landau simulation of partition functions~\cite{vogel_generic_2013}. The
Wang-Landau approach allows us to estimate partition values at zero temperature, whereas the FKT
algorithm can only approach this limit. Another reason for providing a Wang-Landau implementation is
that it can be easily adapted to other codes and noise models, in contrast to the FKT algorithm.
However, we expect that tensor network methods are likely to be a more robust choice \edit{for
settings such as circuit level noise~\cite{piveteau2024}. Further discussion comparing different methods for estimating partition functions can be found in Section~\ref{Discussion}.}
 We provide a brief description of the
Fisher-Kasteleyn-Temperley (FKT) and Wang-Landau (WL) algorithm in Appendix~\ref{FKT_appendix}
and~\ref{WL_appendix}. For an in-depth analysis of these methods, we refer the reader
to~\cite{kasteleyn_statistics_1961, temperley_dimer_1961, thomas_exact_2009,
thomas_numerically_2013, wang_efficient_2001, vogel_generic_2013, vogel_practical_2018}. 
The source code of our implementation is publicly available at~\cite{repo:fkt, repo:wl, repo:postprocessing}.

In Fig.~\ref{fig:0TlimitCombined}  we show a comparison between code performance estimates at zero
temperature generated by the WL algorithm and performance estimates at $T=0.1T_{\text{Nish}}$ generated by
the FKT algorithm. We see agreement of results within error bars. Further decrease of temperature
within the FKT algorithm comes at the cost of significant increase of required bits of precision. A
comparison between FKT results at $T=0.1T_{\text{Nish}}$ and FKT results at $T=0.01T_{\text{Nish}}$ is shown in
Appendix~\ref{ZeroTLimit}, Fig.~\ref{fig:fkt0Tlimit}, and produces close matching between the
logical failure curves.
Thus, estimating partition functions with the FKT algorithm at $T=0.1T_{\text{Nish}}$ indicates to be
sufficient to approximate the zero temperature limit with reasonable bits of precision.
Additionally, close matching between WL performance estimates at zero temperature and at
$T=0.1T_{\text{Nish}}$, shown in Appendix~\ref{ZeroTLimit}, Fig.~\ref{fig:wl0Tlimit}, indicates that
$T=0.1T_{\text{Nish}}$ is a close enough approximation of $T=0$ for the quantities under consideration.
Having established this, all subsequent partition function computations are performed using the FKT
algorithm, with $T=0.1T_{\text{Nish}}$ as a stand-in for zero temperature.

In what follows, we consider $10^{4}$ disorder realizations (samples) for $P\in[ 0.06, 0.12]$. For
each sample, we check whether or not the decoders in Definitions~\ref{maximum_Z_decoder_definition}
and~\ref{probabilistic_Z_decoder_definition} fail (``estimation by counting''), and compute the
ratios in Definitions~\ref{decoding_probability_definition} and~\ref{order_probability_definition}
(``estimation by ratio''). We compute error bars using bootstrapping~\cite{efron1987}, in order to
have a method that works for both estimation methods. We are using $1000$ resamples and $95\%$
confidence level for error estimates by bootstrapping. For the estimation by counting we also
compute error bars using a posterior beta distribution with the Jeffreys prior at $95\%$ confidence
level, and find that these agree with the error bars computed via bootstrapping.

In Fig.~\ref{fig:TNish_FKT} we demonstrate Propositions~\ref{decoding_proposition}
and~\ref{order_proposition} at $T=T_{\text{Nish}}$, and in Fig.~\ref{fig:lowT_FKT} at $T=0.1T_{\text{Nish}}$. In
both cases we see agreement within error bars between estimation by counting and by ratio.

Comparing the crossing points seen at $T=T_{\text{Nish}}$ to those seen at $T=0.1T_{\text{Nish}}$, which in the
large distance limit measure the thresholds of the different decoders, indicates that not only the
phase boundary but also the decodability boundary is reentrant. This shows that in the toric code
under bitflip noise, the criterion for Corollary~\ref{optimality_proposition} is not fulfilled.

Comparing the two decoding strategies at $T=0.1T_{\text{Nish}}$ shows the performance difference between MP
and dMP decoding. We find that the logical failure rates are almost identical at odd distances,
while dMP decoding performs better than MP decoding at small, even distances. The dependence on
parity at zero temperature will be discussed in more detail in the next section. At $T=T_{\text{Nish}}$, we
find that maximum partition function decoding generally outperforms probabilistic partition function
decoding for both even and odd distance, although the threshold is not visibly affected.

Moreover, by comparing decoding probability estimates at $T=T_{\text{Nish}}$ to those at $T=0.1T_{\text{Nish}}$,
which is shown in Fig.~\ref{fig:MLvsdMP}, we find that dMP decoding closely matches the ML decoding
performance for low code distances or low error rates. ML decoding increasingly outperforms dMP
decoding with increasing code distances and error rates.

The size of the error bars are not visible in Fig.~\ref{fig:TNish_FKT} and Fig.~\ref{fig:lowT_FKT},
but are as expected smaller for the estimation by ratio. To illustrate the improved sample
efficiency in estimation by ratio compared to estimation by counting, we estimate the fraction of
the sample size that is sufficient to obtain the same confidence interval. 
\edit{An example of the confidence interval sizes as a function of sample fraction is shown in Fig.~\ref{fig:confidence_intervals} (left) for $p=0.01,d=5$. In this example, less than $3\%$ of the sample size is required to reach as precise of an estimate with the ratio method as when using the full sample size and the counting method. Here, the full sample size is just below two million.\footnote{\edit{A few jobs were canceled due to competing jobs from other users, giving a total sample size of 1999131. Overall, for the confidence interval estimates within this section, the sample sizes vary for different $p$ and $d$. In all cases, saturation of the relative required sample size was observed, providing a sample size independent measure.}}} 
\edit{More generally, the reduction in number of shots needed depends on the desired confidence interval width, as well as on $p$ and $d$. Assuming a fixed target confidence interval width, for instance determined by requiring that the logical error curves of two different decoders can be distinguished unless they agree up to a certain precision, we compare the required number of shots when estimating the logical error rates by counting and by ratio, respectively. In the counting method, the width of the confidence interval as a function of the number of shots can be found by fitting an ansatz based on a posterior beta distribution with the Jeffreys prior, taking the failure fraction at a given $p,d$ as the fitting parameter. Meanwhile, for the ratio method, we find that the confidence interval width differs from that of the counting method by a scale factor $1/r(p,d)$, independent of the number of samples. An example is shown in Appendix~\ref{app:CIratio}.  
The relative required sample size $R$, which we define through the comparison of shots $N$ needed to reach a target confidence interval width $\Delta$,
\begin{equation}
    R = \frac{N_{\rm{ratio}}(\Delta)}{N_{\rm{counting}}(\Delta)},
\end{equation}
}\edit{saturates as $N \to \infty, \Delta \to 0$, with the saturation value depending on $p$ and $d$. The saturation is shown for the example of $p=0.01,d=5$ in Fig.~\ref{fig:confidence_intervals} (right).  The saturation value can be straightforwardly determined from $r(p,d)$ by observing that in the limit $N\to\infty$, the width of the confidence interval scales as $1/\sqrt{N}$. The saturation value of the relative required sample size is therefore given by
\begin{equation}
    R \xrightarrow[N \to \infty]{}  1/r(p,d)^2.
\end{equation}
Interestingly, we observe that in the toric code under bitflip noise, the value of $r(p,d)$ increases in the regimes of both large $p$ and low $p$, as compared to values closer to the threshold.  This is shown in Fig.~\ref{fig:ratio_p_dependency}. As a consequence, the ratio method becomes increasingly favorable in the low noise regime,  where regular sampling struggles.}

\begin{figure}
    \centering
    \includegraphics[width=\textwidth]{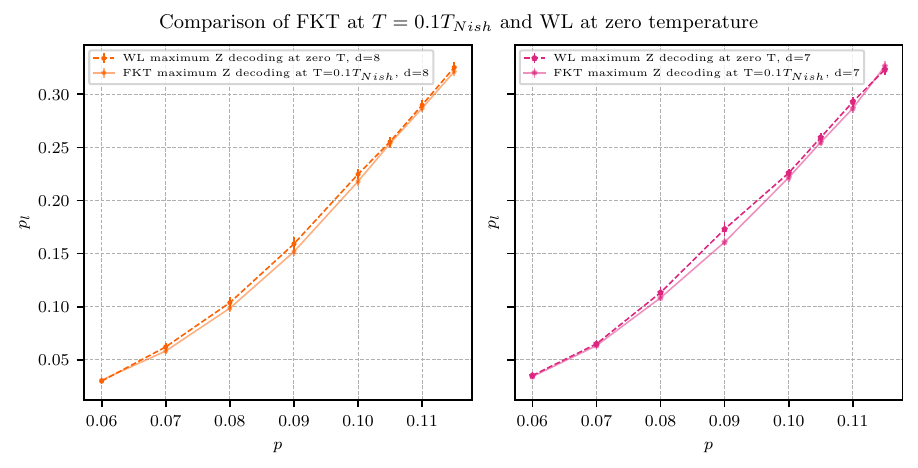}
    \caption{Matching within error bars of maximum partition function decoder performance estimates
    generated by WL at zero temperature and FKT at $T=0.1T_{\text{Nish}}$.}
    \label{fig:0TlimitCombined}
\end{figure}

\begin{figure}
    \centering
    \includegraphics[width=\textwidth]{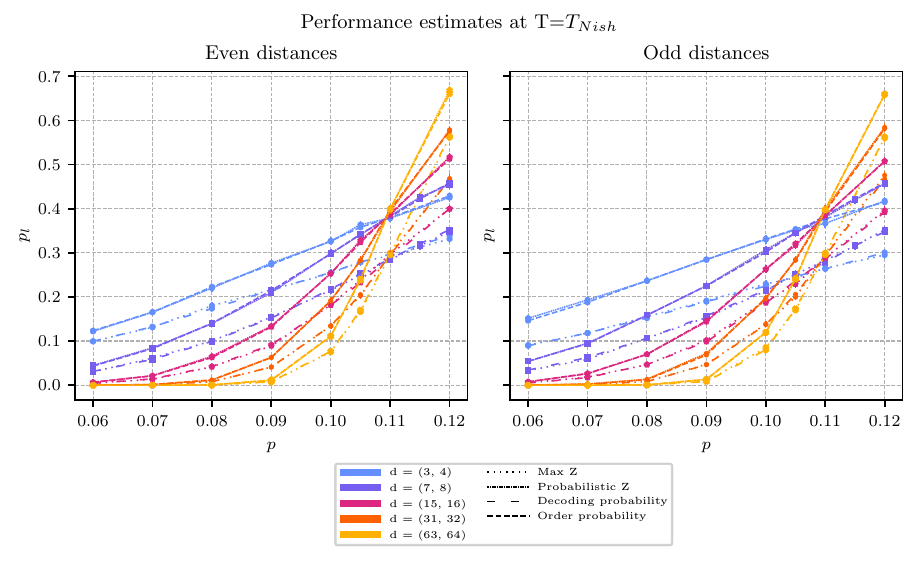}
    \caption{Performance estimates for both counting and ratio methods for probabilistic and maximum
    partition function decoding at $T=T_{\text{Nish}}$. Here, the performance estimates by counting are
    labeled ``Max Z'' and ``Probabilistic Z''. The success rate of ML decoding is measured by
    ``Max Z'' and the ``Decoding probability'' at $T=T_{\text{Nish}}$.
    We find matching within error bars between counting and ratio methods. Furthermore, we find
    improved performance of maximum partition function decoding over probabilistic partition function decoding.}
    \label{fig:TNish_FKT}
\end{figure}

\begin{figure}
    \centering
    \includegraphics[width=\textwidth]{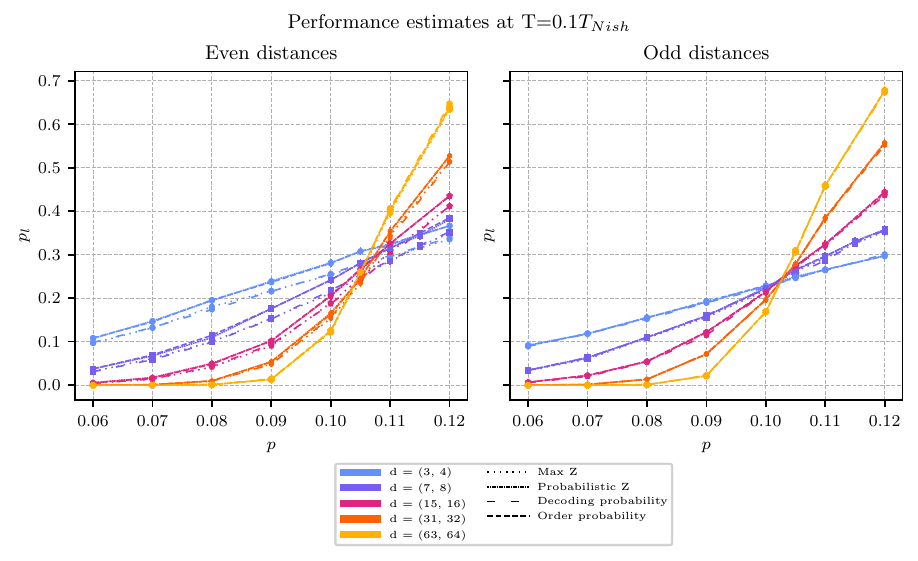}
    \caption{Performance estimates for both counting and ratio methods for probabilistic and maximum
    partition function decoding at $T=0.1T_{\text{Nish}}$. Here, the performance estimates by counting are
    labeled ``Max Z'' and ``Probabilistic Z''. The success rate of MP decoding is estimated by
    ``Probabilistic Z'' decoding and the ``Order probability'' while the success rate of dMP decoding is
    measured by ``Max Z'' and the ``Decoding probability'' in the low temperature limit. We find
    matching within error bars between counting and ratio methods. Furthermore, we find improved
    performance of maximum partition function decoding over probabilistic partition function decoding  for low, even
    distances, while the performance results match closely for odd distances.}
    \label{fig:lowT_FKT}
\end{figure}

\begin{figure}
    \centering
    \includegraphics[width=\textwidth]{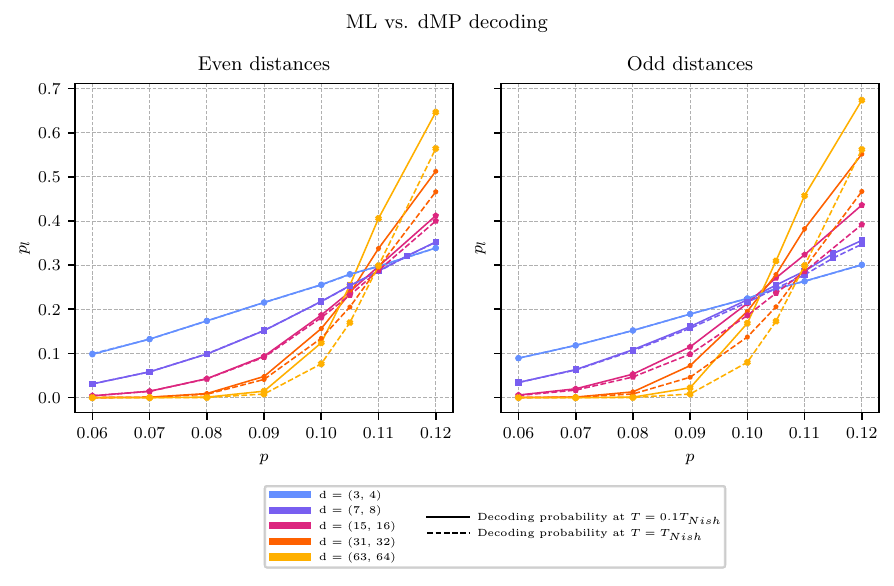}
    \caption{Comparison of dMP (estimated by decoding probability at $T=0.1T_{\text{Nish}}$) and ML
    decoding (estimated by decoding probability at $T=T_{\text{Nish}}$). We see an increase of
    performance advantage from ML over dMP decoding with an increase of code distance.}
    \label{fig:MLvsdMP}
\end{figure}

\begin{figure}
    \centering
    \includegraphics[width=\linewidth]{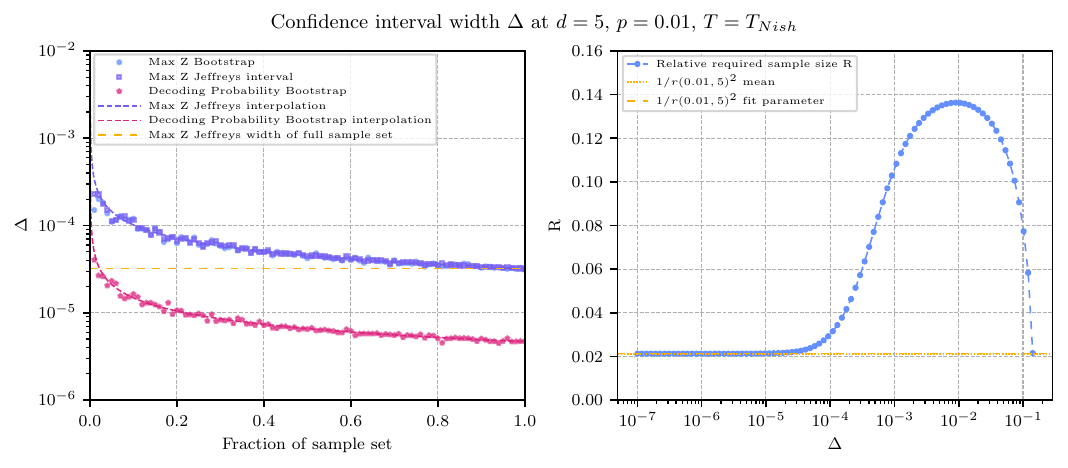}
    \caption{
    \edit{The left plot shows the confidence interval widths for estimates of the optimal logical error rates of the toric code under bitflip noise at $d=5, p=0.01$, for varying fractions of the sample set. The confidence interval widths are displayed for both the decoding probability (``ratio method'') and the maximum partition function counting method. 
    The counting method confidence interval width for the full sample set, generated by the Jeffreys interval, is shown as a horizontal line.
    Its intersection point with the curve corresponding to the ratio method indicates that $2.23\%$ of the sample set is sufficient for estimating the optimal code performance to the same accuracy with the ratio method. The right plot shows the dependency of the relative required sample size $R$ on the target confidence interval width $\Delta$. We find that $R \to 1/r(0.01,5)^{2}$ as $\Delta \to 0$. The value of $r(0.01,5)$ is independent of sample size,  as shown in Fig.~\ref{fig:CIratio}, and is estimated from Fig.~\ref{fig:CIratio} to be $r(0.01,5)\sim 6.87$. This value is compatible with the value of $r(0.01,5)\sim 6.86$ that is obtained by fitting the ratio method data points in the left plot.}}
    \label{fig:confidence_intervals}
\end{figure}

\begin{figure}
    \centering
    \includegraphics[width=0.5\linewidth]{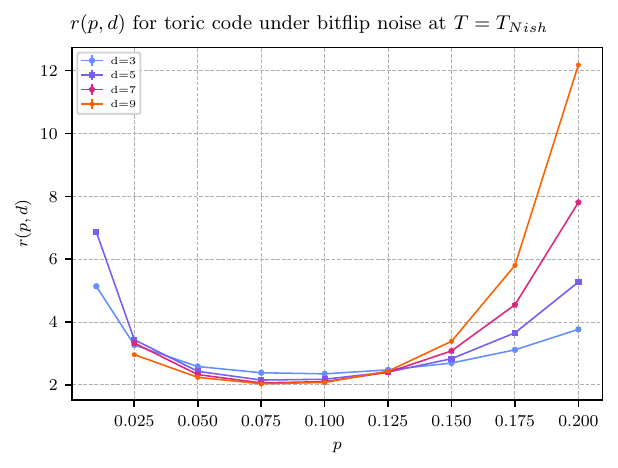}
    \caption{\edit{Scale factor $r(p,d)$ between the confidence interval widths of the counting and ratio methods for estimating optimal logical error rates, for different $p,d$. The scale factor is estimated as in Fig.~\ref{fig:CIratio}. 
    In the example of $d=5,p=0.01$, $r(0.01,5)\sim 6.87$ is compatible with the observed relative required sample size of $R=0.0223$ in Fig.~\ref{fig:confidence_intervals}. The plot only includes data points for $p,d$ where at least $10^2$ logical failures were observed with the counting method, excluding distances $d>5$ at the leftmost point $p=0.01$.}}
    \label{fig:ratio_p_dependency}
\end{figure}

\section{Bias, degeneracy and ensembling in matching based decoding}%
\label{BiasDegeneracyEnsemblingMatching}

In this section, we consider how MP and dMP decoders relate to decoding strategies based on minimum
weight perfect matching (MWPM). While a matching based decoder will return \emph{some} maximum
probability error, it does not necessarily return such errors in an unbiased
fashion~\cite{stace2010}. Comparing MP to MWPM shows how bias in a given matching algorithm affects
the performance.

To approximate dMP decoding, a matching based decoder moreover needs to return not only a single
matching, but an unbiased sampling of several such matchings that can be used to estimate the
relative fraction of them belonging to each class. Methods such as
\emph{ensembling}~\cite{shutty2024} can return several matchings by calling the matching algorithm
multiple times, each time with a random perturbation of the edge weights. We distinguish between
weak ensembling, where the perturbation is small enough that we only sample among minimum weight
matchings, and strong ensembling, where the perturbation is strong enough that we also sample among
higher weight matchings. To approximate dMP decoding, weak ensembling would be used. However, bias
may again affect performance, as ensembling does not return matchings in an unbiased fashion.  An
intuitive example for why this is the case is shown in Fig.~\ref{fig:snake_and_tadpole}. The
comparison of dMP decoding to weakly ensembled MWPM shows how this bias affects the performance. We
also discuss under what conditions weak ensembling would either not be expected to lead to
performance gains, or only lead to performance gains at small distances.

\subsection{Matching based decoding and ensembling}%
\label{section:modifiedMatchingBasedDecoder}

Minimum-weight perfect matching (MWPM) decoding is a maximum probability decoding technique that
finds a most probable correction operator consistent with stabilizer measurement outcomes. We
provide a brief overview of the method in Appendix~\ref{app:mwpm_decoder}.  Implementations of MWPM
decoding, such as sparse blossom~\cite{higgott2025} or fusion blossom~\cite{wu2023}, operate
deterministically. That is, given a matching graph and syndrome, they will always find the exact
same correction operator. Under a uniform bitflip noise model with error probability $p$, we employ
an ensembling~\cite{shutty2024} strategy by sampling $50$ matching graphs with their edge weights
perturbed to
\begin{equation}
    \label{eq:ensemblingPerturbation}
    w_i = \log\frac{1-(p+\xi_i)}{p+\xi_i},
\end{equation}
where $\xi_i\sim\mathcal{N}(0,\sigma_\xi^2)$. If a syndrome $\vec{s}$ is consistent with multiple
error configurations, this small, random modification of the edge weights allows the decoder to find
different correction operators for the different perturbed graphs.  We group the correction
operators into equivalence classes and select a representative correction operator for the
equivalence class found most often as the decoding result. The expectation is that this method on
average will decode to error classes with higher degeneracy, thus approximating dMP decoding,
provided $\sigma_\xi$ is chosen small enough to not introduce new shortest paths into the matching
graph.  As an alternative ensembling strategy, we also tried sampling isomorphic permutations of the
matching graph $\mathcal{G}_M$, where each permutation changes the order of nodes and edges. With a
deterministic decoder implementation, different permutations may result in different equivalent
matchings to be found. While this method produced slight improvements on the observed logical error
rate over MWPM decoding, it generally performed worse than the ensembling based on edge weight
perturbations discussed here.

In the following subsections we present and discuss the results of simulations with a range of
physical error rates for the toric code, as well as the unrotated and rotated planar surface code.
We set $\sigma_\xi=10^{-6}$, after optimizing this parameter as discussed in
Appendix~\ref{app:opt_ensembling}. Our implementation of ensembling, based on \texttt{PyMatching}, is available at~\cite{repo:mwpm}.

\begin{figure}
    \centering
    \includegraphics[width=0.4\linewidth]{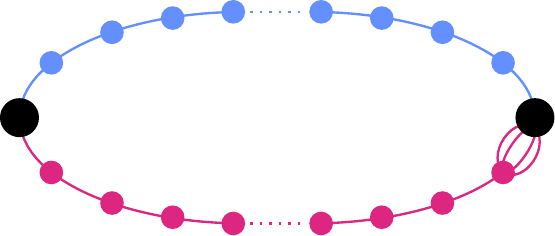}
    \caption{
    Depicted are five shortest paths between the black vertices, corresponding to minimum weight
    perfect matchings. The equivalence classes are indicated by color, with the upper (blue) path
    belonging to one class and the four lower (red) paths belonging to another. The red paths differ
    only by the last segment. A fair sampling of minimum weight perfect matchings would return the
    blue path 20\% of the time. In weak ensembling, the weight of each segment of the above paths is
    perturbed by a small amount. If the perturbations add up in such a way that the blue path is the
    shortest, this is the matching returned. However, given a large number $N_S$ of segments, the
    length difference between the blue path and any of the red paths is almost always determined by
    the first $N_S -1$ segments, so that as $N_S \to \infty$ the blue path is returned $50\%$ of the
    time. This demonstrates that weak ensembling does not sample fairly among matchings.}
    \label{fig:snake_and_tadpole}
\end{figure}

\subsection{Comparison of MP and dMP to matching based decoding in the toric code}%
\label{sect:partition_vs_matching_decoders}

We have established in Propositions~\ref{decoder_statements},~\ref{decoding_proposition}
and~\ref{order_proposition} that the order probability at $T=0$ estimates the success rate of a MP
decoder while the decoding probability at $T=0$ estimates the success rate of a dMP decoder. (In
practice, we use $T=0.1T_{\text{Nish}}$ as a stand-in for $T=0$, as discussed in
Section~\ref{FKT_results_section}). By comparison of logical failure rates from MWPM and weakly
ensembled MWPM decoding to the order probability at $T=0$ and the decoding probability at $T=0$
respectively, one can compare potentially biased matching decoders to their unbiased counterparts.
We note that the choice of normally distributed perturbations $\xi_{i}$ will on rare occasions lead
to sampling of higher weight matchings, but that such events are rare enough to not affect the
results significantly.

The logical performance estimates under MWPM, weakly ensembled MWPM, MP and dMP decoding are shown
for even and odd code distances in Fig.~\ref{fig:mwpmvsmp}. We observe that MWPM decoding suffers
from bias, which primarily affects the performance at even code distances negatively. In particular,
the logical failure rates of the MWPM decoder are noticeably higher than the failure rates of the MP
decoder at even distances beyond $d=4$. In contrast, for odd distances, the logical failure rates of
MWPM and MP decoding differ only slightly.  This parity dependent effect of bias on MWPM decoding is
closely related to the appearance of ground state degeneracies, which are more likely for even
distance toric codes as discussed in Section~\ref{sec:boundaryAndParity}.  Meanwhile, we find that
the code performance of the dMP decoder aligns closely with the performance estimates of weakly
ensembled MWPM decoding. Thus, the bias of weakly ensembled MWPM decoding does not visibly affect
its performance. This may be explained by the ensembling having a certain degree of robustness to
bias:  it requires solely the estimation of which error classes have the highest degeneracy, rather
than an estimation of their exact amounts of degeneracy.

\begin{figure}
    \centering
    \includegraphics[width=\textwidth]{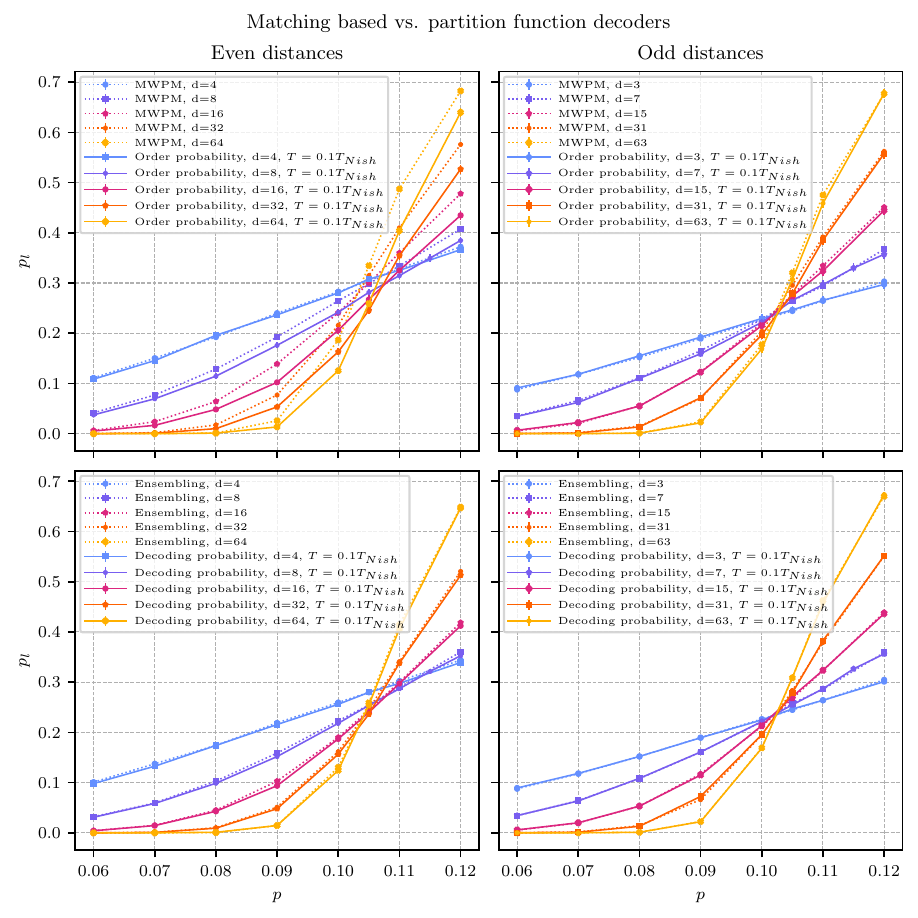}
    \caption{Comparison of potentially biased matching based decoders to their unbiased partition
    function decoder counterparts. In the absence of bias, an MWPM decoder is expected to act as an
    MP decoder while a weakly ensembled MWPM decoder is expected to behave as a dMP decoder. We notice a
    bias-induced deviation between decoding schemes primarily between MWPM and MP decoding at large even
    distances.}
    \label{fig:mwpmvsmp}
\end{figure}

\subsection{Boundary and parity dependence of performance effects from bias and weak ensembling}%
\label{sec:boundaryAndParity}

We have seen in Section~\ref{FKT_results_section} that the difference between MP and dMP depends on
the parity of the distance in the toric code. For degeneracy to matter at zero temperature, there
must be syndromes $\vec{s}$ such that there are at least two different error classes
$\errorclass^{(1)}_{\vec{s}}, \errorclass^{(2)}_{\vec{s}}$ fulfilling the criteria
$n_{\max}(\errorclass^{(1)}_{\vec{s}})>0, n_{\max}(\errorclass^{(2)}_{\vec{s}})>0$ and
$n_{\max}(\errorclass^{(1)}_{\vec{s}})\neq n_{\max}(\errorclass^{(2)}_{\vec{s}})$. In terms of
matchings, the condition is that there are syndromes consistent with at least two classes of minimum
weight perfect matchings, with one class containing more such matchings than the other. We
illustrate such a scenario for the $d=6$ toric code in Fig.~\ref{fig:6x6_torus_degen}. In addition
to governing the difference between MP and dMP, effects from bias in matching based decoders can
only occur when the first of these criteria is fulfilled.

Under the assumption of uniform bitflip noise, the existence of two classes such that
$n_{\max}(\errorclass^{(1)}_{\vec{s}})>0, n_{\max}(\errorclass^{(2)}_{\vec{s}})>0$ can be related to
the existence of even (Hamming) weight representatives of logical operators. Taking such a
representative of weight $2w$, we may write it as the disjoint union of two errors $e^{(1)},
e^{(2)}$ with the same syndrome, with $w(e^{(1)}) = w(e^{(2)}) = w$. Conversely, taking
$e^{(1)}\in\errorclass_{\vec{s}}^{(1)}$ and $e^{(2)}\in\errorclass_{\vec{s}}^{(2)}$ to be
inequivalent maximum probability errors of weights $w(e^{(1)}) = w(e^{(2)})$, there exists a logical
operator $e^{(1)}e^{(2)}$ with weight $2w(e^{(1)}) -  2w(e^{(1)} \cap e^{(2)})$, with $e^{(1)}\cap
e^{(2)}$ being the overlap of the two errors.  The lowest-weight even representative of a logical
operator thus provides a lower bound on how significant the effect of degeneracy can be, as it
lower-bounds the weight of errors that can give rise to syndromes $\vec{s}$ fulfilling the two
criteria above.

In the toric code at even distance, this lower bound allows for effects from degeneracy in the
leading term of the logical error rate, as the lowest-weight even logical representative is of
weight $d$. Considering also the criterion  $n_{\max}^{(1)}\neq n_{\max}^{(2)}$, however, shows that
degeneracy can only show up in the first subleading term.  A maximum probability error chain belongs
to a set of equivalent error chains of equal weight ($n_{\max}>1$) if and only if it can be deformed by
stabilizer multiplication without changing its weight. On the square lattice, it is clear by
inspection that this is only possible when the chain does not consist of only horizontal bonds or
only vertical bonds, which excludes the leading term contribution. Degeneracy only affects the first
subleading term, as illustrated in Fig.~\ref{fig:6x6_torus_degen}.  In the toric code at odd
distance, meanwhile, the effect of degeneracy is increasingly suppressed with $d$, since the
lowest-weight even logical representative is of weight $2d$ (crossing the torus ``diagonally''), as
illustrated in Fig.~\ref{fig:code_ambiguities} (a). From these considerations, we expect the
performance improvement from weak ensembling to be very suppressed in the toric code at odd
distance.

Extending the above considerations to the surface code, we expect improved logical performance from
weak ensembling in the unrotated surface code at both even and odd distances, but only at even
distances in the rotated surface code.  The unrotated surface code has stabilizers of odd weight
(the weight-three stabilizers at the boundary), so that both the odd and even distance code contain
low-weight logical representatives of even length. In Fig.~\ref{fig:code_ambiguities} (b) we
illustrate a syndrome in the unrotated surface code that fulfills both of the above criteria, with
leading-term effects on the logical performance. The rotated surface code contains only even-weight
stabilizers, and all logical representatives have the same parity. At odd distance, we therefore
expect that ensemling has no effect at all, while at even distance we again expect leading-term
effects, as illustrated in Fig.~\ref{fig:code_ambiguities} (c).

In Fig.~\ref{fig:mwpm_uniform_pert}, we show the effect of ensembling on logical performance for
even and odd distance in the toric code, unrotated surface code and rotated surface code. We see
that ensembling does not improve performance in the rotated surface code at odd distance, and barely
differs from MWPM in the toric code at odd distance, while for the other combinations of boundaries
and parities there is a noticeable (though modest) improvement. It should be noted that we have seen
in Section~\ref{sect:partition_vs_matching_decoders} that bias lowers the performance of MWPM
decoding compared to MP decoding. The performance gains from ensembling seen in the surface code
might similarly stem from ensembling having robustness to bias, rather than from the degeneracy
enhancement itself. As noted above, the effects of bias have a similar dependence on boundary
conditions and the parity of the distance as the effects of degeneracy.

\begin{figure}
    \centering
    \subcaptionbox{Defect pattern.\label{sfig:6x6_torus_degen_syn}}
        {\includegraphics{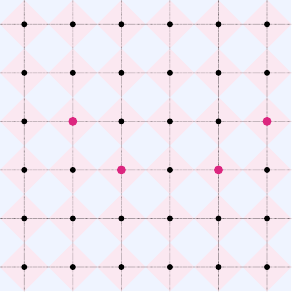}}%
    \hfill
    \subcaptionbox{Matching with degeneracy 1.\label{sfig:6x6_torus_degen_match1}}
        {\includegraphics{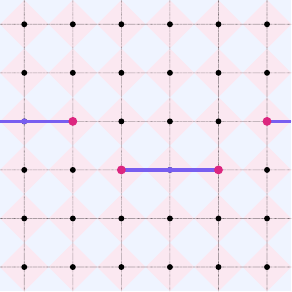}}%
    \hfill
    \subcaptionbox{Matching with degeneracy 4.\label{sfig:6x6_torus_degen_match2}}
        {\includegraphics{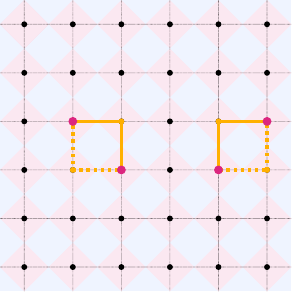}}%
    \caption{In Ref.~\cite{duclos2010} the authors give an example for a syndrome pattern with
    matchings of equal weight but different degeneracies, shown here on the 6$\times$6
    torus~(\subref{sfig:6x6_torus_degen_syn}). On a larger torus, this pattern can be repeated
    arbitrarily often in horizontal direction. The first matching, shown
    in~(\subref{sfig:6x6_torus_degen_match1}) corresponds to an error class with degeneracy 1. The
    syndrome admits another error class with degeneracy 4, shown in the matching
    in~(\subref{sfig:6x6_torus_degen_match2}).}
    \label{fig:6x6_torus_degen}
\end{figure}

\begin{figure}
    \centering
    \subcaptionbox{Distance 5 toric code.\label{sfig:5x5_toric_ineqiv}}
        {\includegraphics{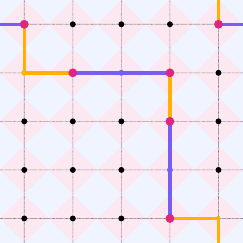}}%
    \hfill
    \subcaptionbox{Distance 5 unrotated surface code.\label{sfig:5x5_planar_ambig}}
        {\includegraphics{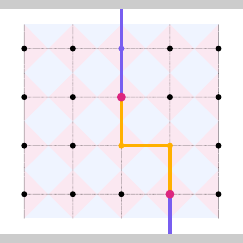}}%
    \hfill
    \subcaptionbox{Distance 6 rotated surface code.\label{sfig:6x6_rotated_ambig}}
        {\includegraphics{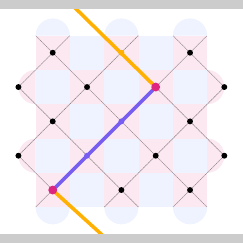}}%
    \caption{Example syndrome pattern for the distance 5 toric code with a chain of two possible
        matchings with equal weight, wrapping around the torus in both
        directions~(\subref{sfig:5x5_toric_ineqiv}). The matching shown in purple has degeneracy 1,
        whereas the matching colored orange has four equivalent alternatives. The unrotated planar
        surface codes admits such ambiguities as well~(\subref{sfig:5x5_planar_ambig}), here shown
        with a matching of degeneracy 1 (purple) that includes the boundary nodes (gray shaded area) and a
        matching with degeneracy 3 (orange). For the unrotated surface code, equivalent matchings of
        different degeneracy can only occur for even distance~(\subref{sfig:6x6_rotated_ambig}).}
    \label{fig:code_ambiguities}
\end{figure}

\begin{figure}
    \centering
    \includegraphics{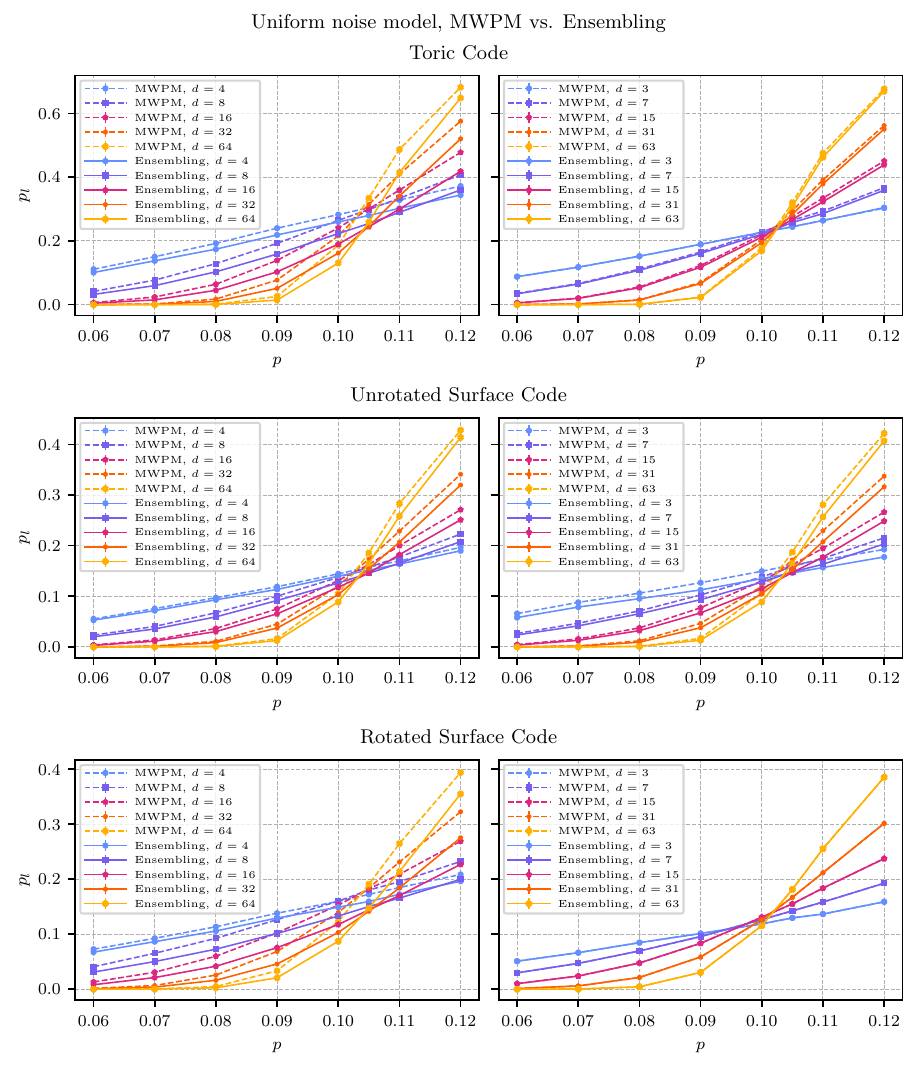}
    \caption{Comparison of  MWPM decoding with ensembling using 50 random perturbations on the
    decoder graph for the toric code, unrotated surface code and rotated planar surface code under a
    uniform noise model.  On the left the logical error curves for even distances are shown, on the
    right for odd distances. We find small improvements of the intersection of logical error curves and
    the physical error rates with ensembling for combinations of boundary conditions and distance parity
    that admit degeneracies.}
    \label{fig:mwpm_uniform_pert}
\end{figure}

\edit{\subsection{Degeneracy and bias in other noise models}\label{OtherNoise}}

\edit{
To explore if the effects of degeneracy and bias may be stronger in other noise models, we consider two candidates. For both the candidates we compare the performance under MWPM decoding and weakly ensembled MWPM decoding.

First, we note that the ``diagonal'' syndrome pattern of Fig. \ref{fig:6x6_torus_degen}, which gives rise to a degeneracy of two each time it occurs, is precisely the one caused by an error on the ancilla qubit in the middle of a ``bare ancilla'' syndrome extraction circuit under a certain CNOT ordering. The circuit takes the following form, with the relevant error propagation marked in red:
\begin{equation}
    \includegraphics[width=0.25\linewidth]{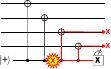}
\end{equation}
We consider a noise model that only includes bitflip errors at this particular point in the bare ancilla syndrome extraction circuit for weight 4 $X$-type stabilizers, corresponding to pairs of bitflip errors on the adjacent data qubits. The CNOT ordering, or equivalently the possible pairs of affected qubits, is chosen such that the distance of the code is not reduced. (The pairs of qubits are aligned orthogonally to the logical bitflip operator $\bar{X}$.) 
We furthermore split the noise model in the construction of the decoding graph, meaning that the MWPM decoder does not take into account the correlations between the errors, but instead performs matching on a decoding graph that corresponds to regular bitflip noise. (Such splitting is performed in matching-based decoders in order to reduce a hypergraph matching problem into a graph matching problem, making it tractable.) In this setting, we find that ensembling has a significant advantage for even distances, as seen in Fig.~\ref{fig:hook_error_plot}. Interestingly, for $d=4$, the observed performance difference is reversed: here, the bias in \texttt{PyMatching} apparently matches the patterns of correlated errors. We show this example separately in Appendix~\ref{app:further_plots}, where we also compare ensembled MWPM decoding to perturbed MWPM decoding (by which we mean considering a single perturbation of the decoding graph). This example illustrates how
biasing a matching-based decoder towards matchings that correspond to known correlations can improve performance.

\begin{figure}
    \centering
    \includegraphics{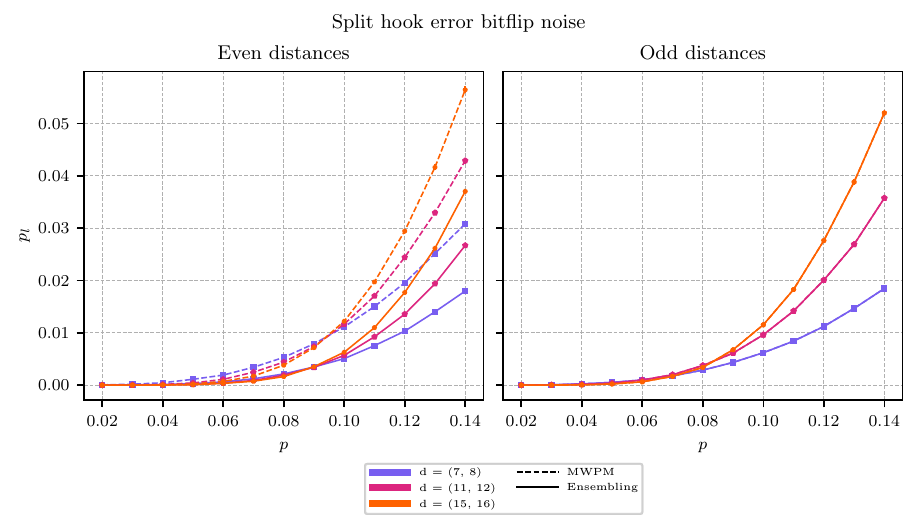}
    \includegraphics{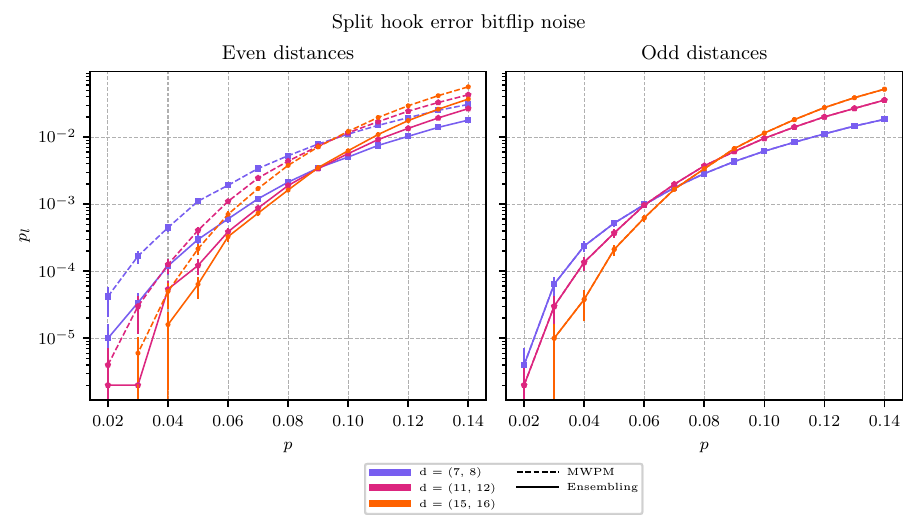}
    \caption{\edit{Top: Logical performance of the surface code under split bitflip hook errors on $X$-type stabilizers for even (left) and odd (right) code distances. Bottom: The same logical performance in log-log-scale, allowing the low-$p$ behavior to be seen. The ensembling uses 50 random perturbations on the decoding graph and improves the performance of MWPM significantly for even code distances, while leaving the performance for odd code distances unchanged.}}
    \label{fig:hook_error_plot}
\end{figure}

Second, we consider the rotated surface code under a combination of bitflip and readout noise, since the 3D setting allows for more same-length paths between endpoints of error chains than the 2D setting. Here, however, we see only a modest performance difference between MWPM decoding and weakly ensembled MWPM decoding. The results are shown in Appendix~\ref{app:further_plots}.}

\FloatBarrier

\section{Strong ensembling, and performance estimates in the toric code under non-uniform bitflip noise}%
\label{StrongEnsembling}

In the toric code under uniform bitflip noise, we have seen in  Section~\ref{FKT_results_section}
that dMP primarily provides benefits over MP for small, even code distances. The performance
difference between dMP and ML decoding shows that there is room for further decoder improvements,
beyond degeneracy enhancement.  To deepen the analysis of decoder optimization potential, we extend
our investigation to a non-uniform noise model.  In this model, each qubit bitflip error rate,
$p_{i}$, is independently drawn from a normal distribution with mean $p$ and standard deviation
$\sigma_p$. The error ratio samples are subsequently truncated to lie within the interval $[10^{-4},
\frac{1}{2}]$ which ensures physical plausibility.  This noise model lifts the ground state
degeneracies and is thus expected to reduce dMP decoding to MP decoding already for a small non-zero
standard deviation.

In the following subsections, we show the results of simulations for MP, dMP, and ML decoders under
the truncated Gaussian noise model for varying standard deviations $\sigma_p \in \{0.005, 0.06,
0.12\}$ with $10^{4}$ samples of error configurations with Gaussian mean physical error rates
$p\in[0.06, 0.14]$.  The energy gap between ground and excited states may be arbitrarily small \edit{for}
non-uniform couplings introduced in Section~\ref{StatMechMapping}. To suppress non groundstate
contributions to performance estimates of MP and dMP decoding sufficiently, we are approximating the
zero temperature limit by reduced temperature $T=0.01T_{\text{Nish}}$ and increased $9999$ bits of
precision within the FKT algorithm for non-uniform bitflip noise.

\subsection{The effect of non-uniformity on the difference between MP and ML performance}
Fig.~\ref{fig:fktNonUniform} presents the performance estimates of MP, dMP, and ML decoders under
the truncated Gaussian noise model.  We observe that increasing the standard deviation generally
improves overall code performance. This is shown by lower absolute logical failure rates and higher
error thresholds at larger standard deviations of the noise model with respect to the same decoding
strategy. This is expected as the decoding process for less uniform systems generally becomes easier
when the overall error expectancy remains the same.  For odd code distances, the relationship
between decoding strategies remains largely unaffected by increased standard deviation: dMP decoding
performs only as good as MP decoding for all investigated physical error rates, while ML decoding
consistently outperforms dMP and MP decoding for larger error rates and code distances $d>3$. Hence,
weak ensembling methods are expected to produce no notable performance gain over MP decoding for odd
code distances under non-uniform bitflip noise. Further, a clear performance gap to ML decoding
persists across all tested standard deviations and code distances $d>3$.  In contrast, for even code
distances, we observe that dMP decoding performs only as good as MP decoding already for small non
uniformity in the noise. This indicates no performance enhancements by weak ensembling over standard
MWPM for slightly non-uniform bitflip noise, as expected from the lifting of the ground state
degeneracies. As a consequence, the ML decoder now outperforms dMP decoding at higher physical error
rates already at small code distances, in contrast to the uniform case. It is important to note that
for small enough physical error rates and high enough standard deviation all decoding strategies
become indistinguishable, while the performance enhancement opportunities between MP/dMP toward ML
decoding remain for high physical error rates within the tested range of standard deviations. The
performance difference between ML and dMP/MP decoding decreases slightly with an increase of
standard deviation.

\begin{figure}
    \centering
    \includegraphics{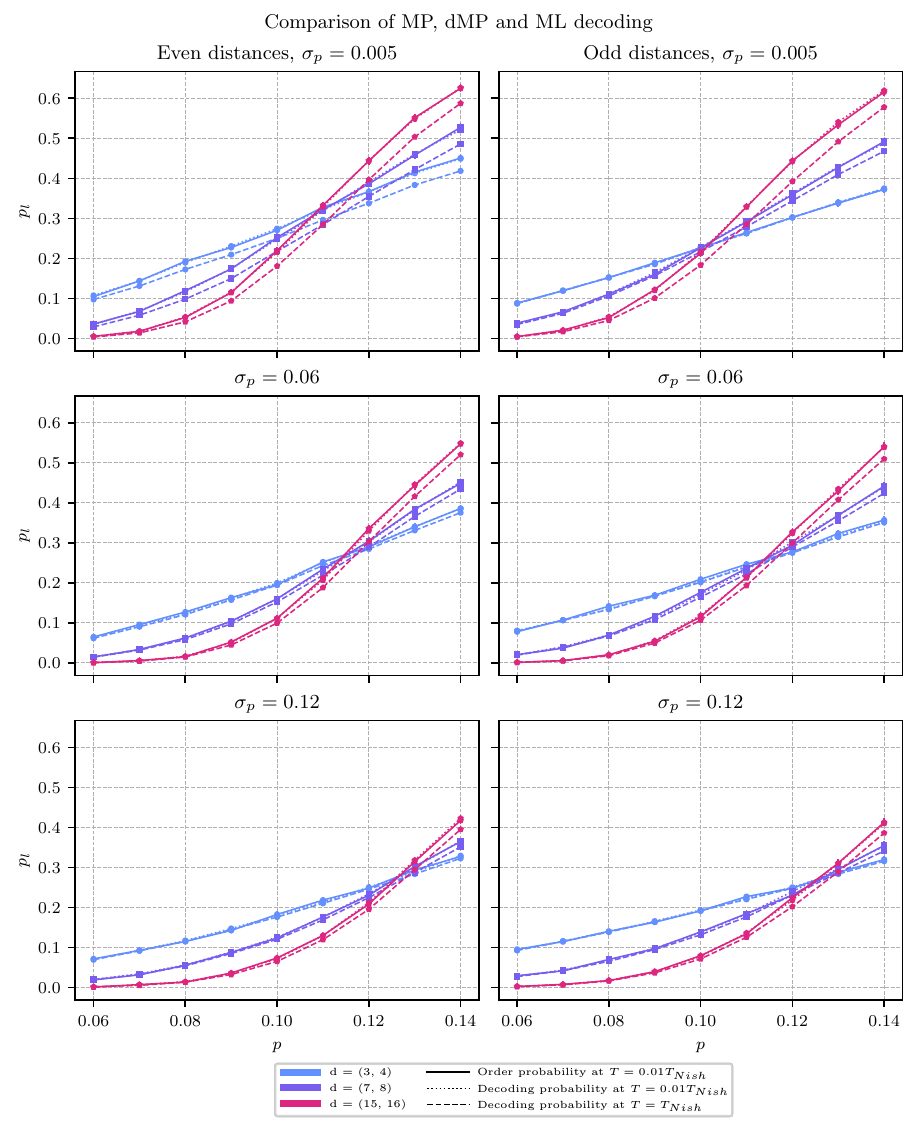}
    \caption{Comparison of MP, dMP, and ML decoding for varying standard deviations of the cut-off
    Gaussian bitflip noise model. We observe general performance improvement with an increase of
    standard deviation. We also see that the performance difference between dMP and ML decoding persists
    within investigated parameter range.}
    \label{fig:fktNonUniform}
\end{figure}

\subsection{Strong ensembling for non-uniform noise}%
\label{ssec:nonuniform_mwpm}

\begin{figure}
    \centering
    \includegraphics{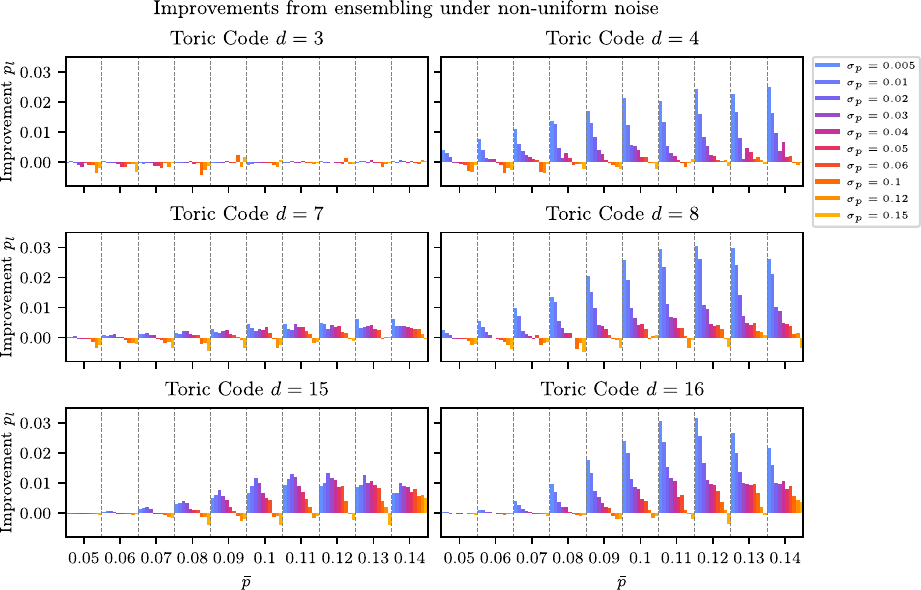}
    \caption{Improvement of the logical error rate under ensembling with $\sigma_\Xi=0.025$ for a
        non-uniform noise model. For odd distances we see no or only very small improvements for
        $d=15$, whereas for small standard deviation $\sigma_p$ of the physical error rate and even
        distances, improvements up to $0.03$ in $p_l$ can be found. This effect decreases with
        larger standard deviations and vanishes for $\sigma_p=0.15$. The improvements are more
        pronounced for physical error rates around the threshold, especially at larger distance.}
    \label{fig:mwpm_pert_nonuniform}
\end{figure}

We evaluated ensembling for the MWPM decoder for the toric code under the same truncated Gaussian
noise model. Similar to the discussion in Section~\ref{section:modifiedMatchingBasedDecoder} we
sample 50 perturbed matching graphs with modified edge weights:
\begin{equation}
    w_i = \log\frac{1-(p_i+\Xi_i)}{p_i+\Xi_i},
\end{equation}
where we sample $\Xi_i\sim\mathcal{N}(0, \sigma_\Xi)$. We again subjected the perturbation standard
deviation $\sigma_\Xi$ to optimization outlined in Appendix~\ref{app:opt_ensembling} and found
$\sigma_\Xi=0.025$ as an approximate optimum for the configurations considered.
Fig.~\ref{fig:mwpm_pert_nonuniform} shows the improvement in the logical error rate $p_l$ for a
range of mean physical error rates, standard deviations for the non-uniform noise model, and
distances $d=3,4,7,8,15,16$ for the toric code. For odd distances, we find no or very small
improvements only, which aligns with the findings for a uniform noise model and the previous
discussions about degeneracies in the toric code.  For even distances, we see improvements in the
logical error rate up to about $0.03$, particularly for mean physical error rates around the
threshold. The effect is most pronounced for small $\sigma_p$ and vanishes for truncated Gaussian
noise with a standard deviation for the physical error rate of $\sigma_p=0.15$.

\edit{
\section{The color code under bitflip noise and depolarizing noise}\label{ColorCode}
As a first exploration into applying the partition function framework to other codes, we consider the color code on a hexagonal lattice under bitflip and depolarizing noise on the data qubits. In these settings, the FKT algorithm cannot be used. Given the long runtimes needed for Wang-Landau sampling in the toric code under bitflip noise, we turn instead to tensor network methods for partition function computations. Restricting to the Nishimori line, we employ the publicly available tensor network contraction package \texttt{SweepContractor.jl}~\cite{chubb2021e}.  This algorithm has shown reliable results for approximate maximum likelihood decoding of the color code under depolarizing noise.

The tensor network subject to the contraction constitutes of marginal error probability tensors associated with each logical qubit, connected to delta tensors for each stabilizer generator~\cite{chubb2021c}. Errors are sampled according to the bitflip or depolarizing channel for varying physical error probabilities, inducing index flips in the marginal distribution tensors of the respective logical qubits. The contraction of the network sums over all error configuration within the same logical class, up to tensor truncations. The contraction acts directly on the code structure, making temperatures other than the Nishimori temperature (and thus the investigation of MP and dMP decoding) inaccessible by this method.

We observed that for given values of $p$ and $d$, an insufficient maximum bond dimension $\chi_{\rm{max}}$ can introduce sign errors through truncation of the tensor network, leading to unphysical (negative) estimates of the logical class probabilities. 
To ensure sufficient precision while keeping the total runtime low, 
we adapted $\chi_{\rm{max}}$ to each combination of $p$ and $d$, finding that the required values varied between $\chi_{\rm{max}} = 32$ and $\chi_{\rm{max}}=128$ for the $p$ and $d$ considered, with the larger values within this range required under depolarizing noise at low $p$ and large $d$.\footnote{\edit{Methods like gauge fixing can also be applied to stabilize signs within the truncation~\cite{bro-sign-flip}, which could potentially lower the requirements on $\chi_{\rm{max}}$. We found that a minimal version of this method, which deterministically assigns the orientation of only the leading singular vector in the singular value decomposition,  was insufficient to change the required $\chi_{\rm{max}}$ in our case, and did not attempt more elaborate versions of the scheme.}}
}

\edit{Using the partition functions computed through tensor network contraction, the maximum partition function decoder is constructed following Def.~\ref{maximum_Z_decoder_definition}, and the decoding probability is evaluated in accordance with Def.~\ref{decoding_probability_definition}, Eq.~\eqref{decoding_probability_definition_eq}. 
In Fig.~\ref{fig:color_code_logicals} the optimal logical error curves are shown for the color code under bitflip and depolarizing noise, computed both using decoding probability (ratio method) and by counting the number of failures of the maximum partition function decoder.

As in the case of the toric code under bitflip noise, we expect the ratio method to require fewer samples in order to achieve a given target confidence interval $\Delta$. In Fig.~\ref{fig:color_code_CI_widths}, examples are shown of the confidence intervals achieved with the two methods. 
The scale factor $r(p,d)$ between the width of the confidence intervals obtained by the counting method and the ratio method again depends on $p$ and $d$, as shown in Fig.~\ref{fig:color_code_CI_ratios}. Example plots demonstrating the independence of $r(p,d)$ from the total sample size are shown in Appendix~\ref{app:CIratio}. We recall from Section~\ref{FKT_results_section} that the relative required sample size to reach a target $\Delta$ is given by $R \to 1/r(p,d)^2$ in the large sample limit. 
While less dramatic efficiency gains are observed under depolarizing noise, where the increase in the required $\chi_{\rm{max}}$ limits the accessible range of $p$, the results point towards an increase in efficiency gain at low $p$ under both bitflip and depolarizing noise.}

\begin{figure}
    \centering
    \text{\scriptsize Performance estimates for the color code at $T=T_{Nish}$}
    \vfill
    \begin{subfigure}{0.49\linewidth}
        \includegraphics[width=\linewidth]{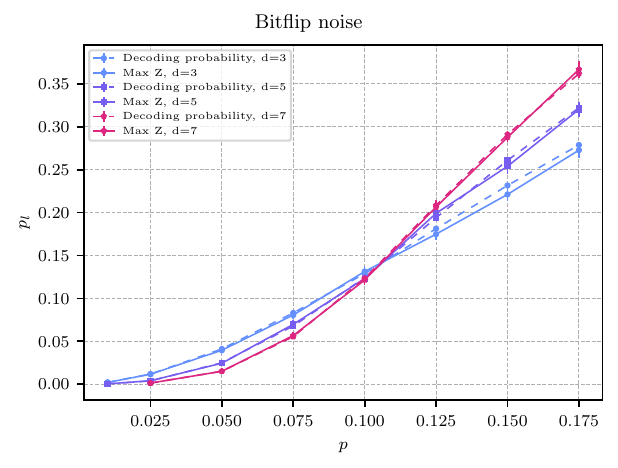}%
    \end{subfigure}
    \begin{subfigure}{0.49\linewidth}
        \includegraphics[width=\linewidth]{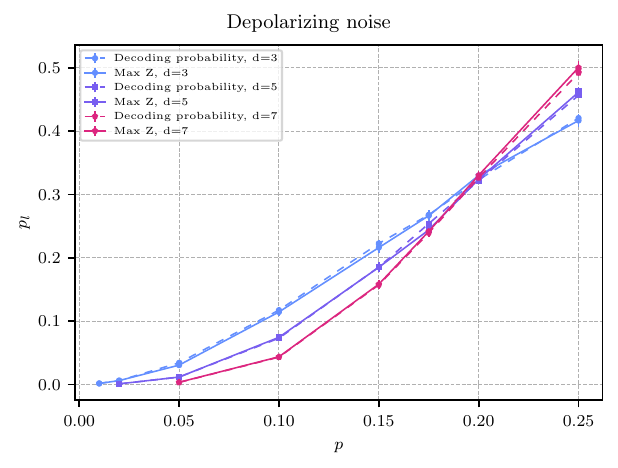}
    \end{subfigure}
    \caption{
    \edit{Estimated optimal performance for the color code under bitflip noise (left) and depolarizing noise (right). The performance estimates obtained from both the counting method (``Max Z'') and ratio method are shown and yield compatible results. The corresponding threshold estimates are consistent with expected values of $p_{th}\sim 10.9\%$ for bitflip noise \cite{PhysRevLett.103.090501} and $p_{th}\sim 18.9\%$ for depolarizing noise \cite{bombin2012a}.}}  
    \label{fig:color_code_logicals}
\end{figure}

\begin{figure}
    \centering
    \text{\scriptsize Confidence interval width $\Delta$ for the color code at $T=T_{Nish}$}
    \vfill
    \begin{subfigure}{0.49\linewidth}    
        \includegraphics[width=\linewidth]{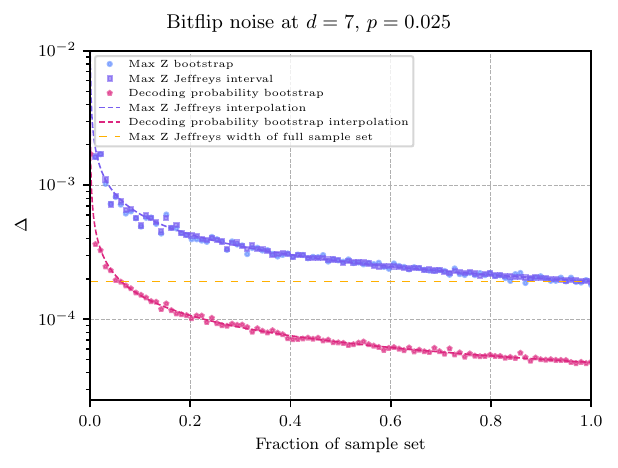}
    \end{subfigure}
    \begin{subfigure}{0.49\linewidth}
        \includegraphics[width=\linewidth]{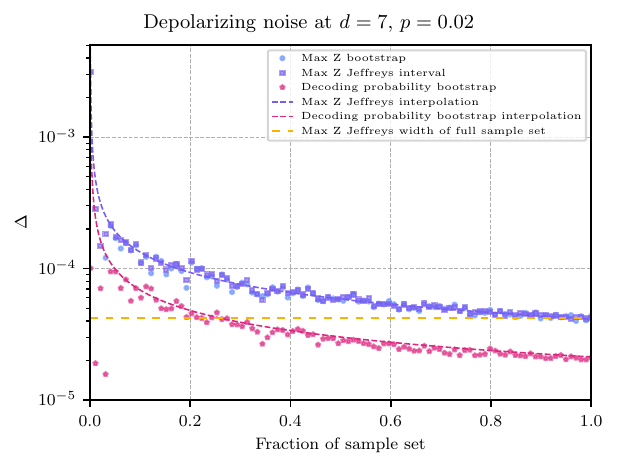}
    \end{subfigure}
    \caption{\edit{Confidence interval widths for estimates of the optimal logical error rates of the color code under bitflip noise at $d=7, p=0.025$ (left) and under depolarizing noise at $d=7,p=0.02$ (right), for varying fractions of the sample sets. The confidence interval widths are displayed for both the decoding probability (``ratio method'') and the maximum partition function counting method. 
    The counting method confidence interval widths for the full sample set, generated by the Jeffreys interval, are shown as horizontal lines.
    Their intersection points with the curves corresponding to the ratio method indicates that under bitflip noise at $d=7,p=0.025$, $6.3\%$ of the sample set is sufficient for estimating the optimal code performance to the same accuracy with the ratio method, while the corresponding number under depolarizing noise at $d=7,p=0.02$ is $26\%$. The full sample sizes are $5\times 10^5$ and $8\times 10^5$, respectively.}}
    \label{fig:color_code_CI_widths}
\end{figure}

\begin{figure}
    \centering
    \text{\scriptsize $r(p,d)$ for the color code at $T=T_{Nish}$}
    \vfill
    \begin{subfigure}{0.49\textwidth}
        \includegraphics[width=\linewidth]{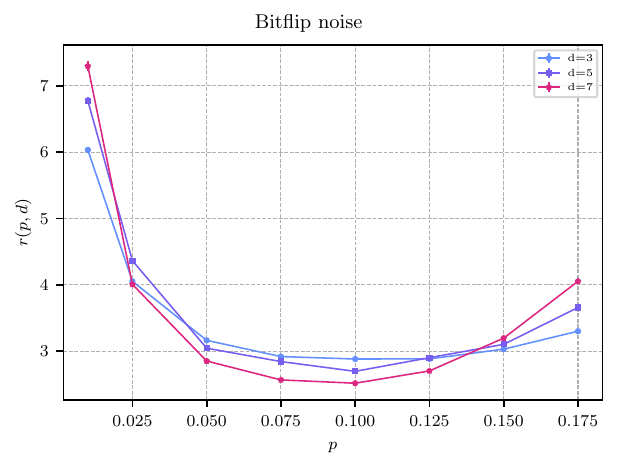}%
    \end{subfigure}
    \hfill
    \begin{subfigure}{0.49\textwidth}
        \includegraphics[width=\linewidth]{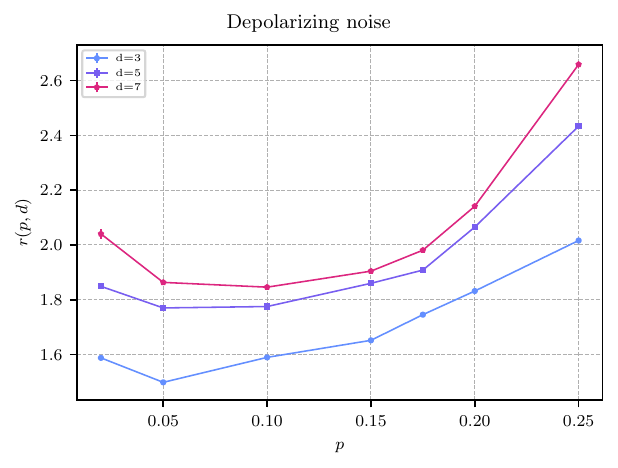}
    \end{subfigure}
    \caption{\edit{Scale factor $r(p,d)$ between the confidence interval widths of the counting and ratio methods for estimating optimal logical error rates of the color code, for different $p,d$. Examples of scale factor estimation for the color code are shown in  Fig.~\ref{fig:CIratio_cc}. The results are shown for the color code under bitflip noise (left) and depolarizing noise (right). Under both noise models, the ratio method exhibits an increasing advantage in the low-$p$ regime.}}
    \label{fig:color_code_CI_ratios}
\end{figure}

\section{Discussion}\label{Discussion}

In this paper, we have presented a framework for estimating logical performance in stabilizer codes 
under different decoding approaches.
The framework relies on estimating partition functions in the corresponding statistical mechanics models.
While the numerical results \edit{mainly focus on} the toric code under bitflip noise, where the partition 
functions can be computed with Pfaffian
methods, the method carries over to more general settings, \edit{such as more general qLDPC codes\footnote{\edit{See e.g. Ref.~\cite{Kovalev:2014vhn}, which leverages the statistical mechanics mapping to explore the decoding transition in non-local, finite-rate qLDPC codes.} } and circuit level noise, where partition functions can be estimated through other means,} using e.g., tensor network methods. 
\edit{We showcase such an application of our framework on top of tensor network methods by considering the color code under bitflip and depolarizing noise, although the distances considered are more limited.

The feasibility of extending to more general codes and noise models depends on the feasibility of computing or approximating partition functions -- or, if restricting to the Nishimori line -- logical error class probabilities. Tensor network methods have shown great promise for 2D problems, while (so far) being far more limited in 3D: in \cite{chubb2021e}, distances up to 96 are reached for the surface code under bitflip noise, while in \cite{piveteau2024}, where circuit level noise is considered for the surface code, the distances are far more modest (up to d=7). Outside of tensor network methods, the Wang-Landau method can also be applied to general settings. 
It is possible that there are regimes in which Wang-Landau could be preferable; we have not seen any systematic findings in the literature on which regimes would be the most interesting, but as far as we can tell it is not known to depend strongly on dimensionality and structure in the same way tensor network methods do, and it would be of interest to explore whether Wang-Landau may outperform tensor network methods for qLDPC codes that lack these features.

Compared to tensor network methods and Wang-Landau, the FKT algorithm used in this work is less likely to be useful for any realistic noise settings. However, in \cite{tobias2026optimaldecodingworm}, a sampling-based algorithm known as the ``worm'' makes it possible to efficiently approximate logical error class probabilities (and, more generally, partition functions at different temperatures) for a slightly broader class of problems: those with a ``matchable'' structure. While still restrictive, this broader class includes the surface code under both bitflip and readout noise. The authors demonstrate how this method can be combined with correlated decoding to improve the surface code depolarizing threshold over correlated matching, exemplifying how methods that can handle only part of the decoding problem optimally may still improve the overall performance.} \edit{A different matching-based method for estimating logical error class probabilities, by finding the first $K$ minimum weight matchings from the decoding graph, is given in \cite{MaoApproximateML}.}

\edit{The partition function framework that we define} 
allows for apples-to-apples comparisons of the logical error curves of different codes, avoiding
confounding factors from code-specific decoder implementations. One example of a confounding factor in the comparison of codes and decoders is bias within matching
implementations. As was pointed out in~\cite{stace2010}, this bias can also lead to
inaccurate estimates of the location of the phase boundary of the corresponding statistical
mechanics model. In Fig.~\ref{fig:mwpm_uniform_pert}, ensembling is seen to outperform MWPM for all
distances, while in Fig.~\ref{fig:lowT_FKT} MP and dMP are seen to agree in the toric code at large
distance. This shows that the improvement in Fig.~\ref{fig:mwpm_uniform_pert} is not due to an
inherent difference between MP and dMP at large distance, but rather due to ensembling suffering
less from bias in the \texttt{PyMatching} implementation (as is also seen in
Fig.~\ref{fig:mwpmvsmp}). 
\edit{Furthermore, the framework is sample-efficient compared to conventional sampling of decoder failures. The sample efficiency depends on $p$ and $d$, and the results indicate that estimation through ratios of partition functions may be particularly 
useful at low $p$ and large $d$, meaning at low logical error rates, where conventional sampling struggles. Probing the regime of low logical error rates is important, as this is the regime relevant for large-scale quantum computing. Other recent work for estimating or bounding the behavior in this regime includes Ref.~\cite{beverland2025failfasttechniquesprobe}.}

The bitflip noise model is far from realistic, and non-uniform noise is furthermore unlikely to be
normally distributed. Future work would be needed to establish which of the qualitative features
seen above would generalize to more realistic settings. For instance, we see that the performance
improves for non-uniform noise when the decoder is given the individual qubit fidelities. For
correlated noise such as circuit noise there could be an overall performance reduction compared to
uniform noise. Similarly, qualitative observations concerning the performance gains from ensembling
(that they are most notable around the threshold, and that they decrease for strongly non-uniform
noise) may change in the presence of correlations.

The role of degeneracy would also merit further work.  For bitflip noise, it was demonstrated
in~\cite{beverland2019} how the increased amount of degeneracy in the rotated surface code can lead
to worse logical performance than that of the unrotated surface code in certain regimes, despite
having a distance that is larger by a factor of $\sqrt{2}$ for the same number of physical qubits.
In~\cite{stace2010} the parity of the distance was seen to clearly influence the threshold
estimates when matching degeneracy is taken into account, in the setting of the toric code under
bitflip noise.  These observations of the effects of boundaries and parity agree with those of the
present work. Future work would be needed to quantify the effect of degeneracy and its dependence on
boundary conditions and parity in more general settings. In stabilizer codes or noise models with
stronger effects from degeneracy, methods such as ensembling could lead to larger gains.
\edit{We see an illustrative example of this in the case of the rotated surface code under a split hook error noise model.}

The numerical simulations in the present work have focused on $T=T_{\text{Nish}}$ and the zero temperature
limit, as these are the temperatures relevant for ML, MP and dMP decoding. Here, we find that the
threshold estimates for maximum partition function decoding and probabilistic partition function
decoding agree. This is consistent with a maximum partition function decodability boundary that is
the same as the phase boundary (although the possibility remains that they may disagree at other
temperatures). We leave as an open question whether there are statistical mechanics models where the
maximum partition function decodability boundary differs from the phase boundary, or whether it is
possible to prove that these boundaries must agree.

\newpage

\section{Author contributions}

LW performed the FKT simulations \edit{and tensor network simulations,} and HH performed the \texttt{PyMatching} simulations. LW and HH
also analyzed the results. LGS conceived the project and developed the theory. EM developed the
initial CUDA code then LW took over. EM and LGS guided the direction of the project, and supervised
the numerical simulations. All authors contributed to the discussion of the results and the writing
of the manuscript.

\section{Acknowledgments}

LGS is supported through a Leverhulme-Peierls Fellowship at the University of Oxford, funded by
grant no. LIP-2020-014. LW, HH and EM acknowledge funding by the German Ministry of Economic
Affairs and Climate Action (BMWK) and the German Aerospace Center (DLR) in project QuDA-KI under
grant no. 50RA2206.

We thank Matthew Steinberg for related discussions and collaboration, Marius Beuerle for initial
collaboration, Benedikt Placke for discussion, and Bela Bauer and Christina Knapp for discussions
about non-uniform noise. We than Steven Simon and Matthew Steinberg for feedback on a draft version of the manuscript.

\appendix

\section{Zero Temperature limit}%
\label{ZeroTLimit}

In order to determine a temperature scale sufficient to estimate decoder performance in the zero
temperature limit with the FKT algorithm, we performed simulations with varying temperature and bits
of precision parameters.  The performance estimates for a selection of these simulations are shown
in Fig.~\ref{fig:fkt0Tlimit} and Fig.~\ref{fig:wl0Tlimit}. Fig.~\ref{fig:fkt0Tlimit} shows a
comparison of FKT performance estimates of maximum partition function decoding at $9192$ bits of
precision for $T=0.01T_{\text{Nish}}$ and $4096$ bits of precision for $T=0.1T_{\text{Nish}}$. The performance
estimates are matching within error bars. Additionally, Fig.~\ref{fig:wl0Tlimit} shows the WL
performance estimates of maximum partition function decoding at $T=0.1T_{\text{Nish}}$ and WL performance
estimates of maximum partition function decoding at $T=0$. The latter reduces to maximizing
$g({E_{\min}(\vec{s})})$, a quantity that WL gives direct access to. The figure shows very close
matching of performance estimates. Fig.~\ref{fig:0TlimitCombined} depicts a further comparison
between WL performance estimates of maximum partition function decoding at $T=0$ and FKT results of
maximum partition function decoding at $T=0.1T_{\text{Nish}}$ with $4096$ bits of precision. As the
performance estimates are matching within error bars, we determine $T=0.1T_{\text{Nish}}$ and $4096$ bits
of precision to be sufficient parameters to estimate the low temperature performance with the FKT
algorithm.

\begin{figure}
    \centering
    \includegraphics{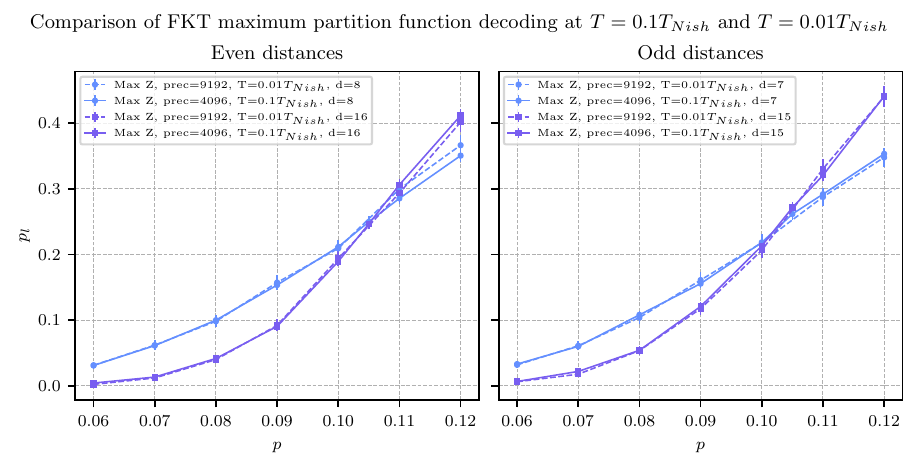}
    \caption{Comparison of FKT performance estimates of maximum partition function decoding at
    $T=0.1T_{\text{Nish}}$ and $T=0.01T_{\text{Nish}}$ for respective number of bits of precision $4096$ and
    $9192$.}
    \label{fig:fkt0Tlimit}
\end{figure}

\begin{figure}
    \centering
    \includegraphics{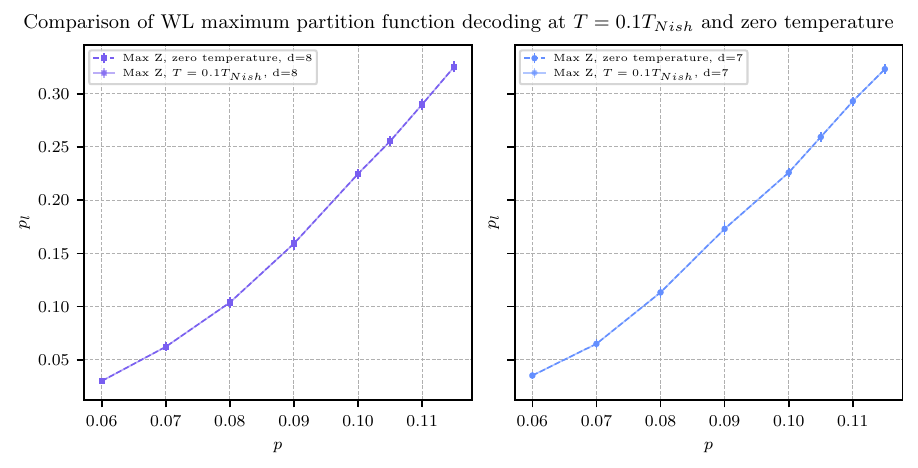}
    \caption{Comparison of WL performance estimates of maximum partition function decoding at
    $T=0.1T_{\text{Nishimori}}$ and $T=0$.}
    \label{fig:wl0Tlimit}
\end{figure}

\section{FKT algorithm}%
\label{FKT_appendix}

The toric code under bitflip noise maps under the statistical mechanics mapping to a two-dimensional
Random Bond Ising Model, as shown in Eq.~\eqref{Hamiltonian_RBIM}. The Ising lattice defines a graph
$G=(V,E)$ by associating Ising spins with vertices and interactions between spins with edges.  The
dual-lattice $G^{\star}$ forms an $m\times n$ square lattice with spin degrees of freedom associated
with faces, interactions transmitted via edges between neighboring faces, and periodic boundary
conditions. A spin configuration on the Ising lattice is described by a face value assignment $\{
S_{i} | S_{i}\in {\pm 1} \land 0\leq i<m\times n \}$ on $G^{\star}$. The correspondence between spin
configurations on $G^{\star}$ and relative domain walls with additional fixing of a single spin
value is depicted in Fig.~\ref{fig:relative_domain_walls}. Furthermore, by decorating each vertex in
$G^{\star}$ with a Kasteleyn city as depicted in Fig.~\ref{fig:dimer_cover}, one realizes a map
between perfect matchings in the decorated graph $\tilde{G}^{\star}$ and spin configurations of the
Ising model. We note that this map is not one-to-one. We endow $\tilde{G}^{\star}$ with an
orientation $D$ and according edge weights by $\omega((i,j))=0$ for $i,j\in \text{vertices of same
Kasteleyn city}$ and $\omega((i,j))=J_{ij}$ for $(i,j) \in \text{D and i, j from different cities}$.
We define the Kasteleyn matrix $K=(K_{ij})_{i,j \in V}$ as the adjacency matrix of
$\tilde{G}^{\star}$, which carries information on connectivity, weights and the orientation of the
graph as follows: $K_{ij}=e^{-2\beta\omega((i,j))}$ if the edge $(i,j)$ is an oriented edge in $D$,
$K_{ij}=-e^{-2\beta\omega((j,i))}$ if $(j, i)$ is an oriented edge in $D$, and $K_{ij}=0$ otherwise.
The skew-symmetric weighted adjacency matrix $K=(K_{ij})_{i,j}$ has dimension $4nm\times4nm$. The
Pfaffian of a $2N \times 2N$ skew-symmetric matrix $K$ is defined by:
\begin{equation}
    \label{eq:pfaffian}
    Pf(K)=\frac{1}{n!2^{n}}\sum_{\pi\in S_{2n}}\sigma(\pi)\prod_{j=1}^{n}K_{\pi(2j-1), \pi(2j)}
\end{equation}
with $\sigma(\pi)$ the sign of the permutation. All non vanishing terms in the Pfaffian correspond
to products of edge weights of a perfect matching on the related graph with signs depending on the
orientation.  For all planar graphs there exists an orientation called Pfaffian orientation which
ensures that all perfect matchings contribute with the same sign to the Pfaffian. One can show for
planar graphs that $Z=Pf(K_{\text{Pf}})$ with the Kasteleyn matrix $K_{\text{Pf}}$ corresponding to
a Pfaffian orientation of $\tilde{G}^{\star}$.  Notice that in the non planar case not all perfect
matchings on $\tilde{G}^{\star}$ are associated with physical spin configurations; domain walls
which form non trivial loops around the torus must come in pairs while dimer configurations can come
in four distinct ways: $(o, o)$, $(o,e)$, $(e,o)$, $(e,e)$. The tuples denote even or odd wrapping
number along the two dimensions of the torus. Only the $(e,e)$ component is related to physical spin
configurations.  To relate the partition function of an Ising model with periodic boundary
conditions to the calculation of Pfaffians one chooses four different orientations on
$\tilde{G}^{\star}$. The four orientations differ only on edges at the boundary (the edges which
wrap around the surface) and are equal to a specific Pfaffian orientation on the bulk. Specifically,
one assigns uniformly for the first row and column either $K_{ij}=e^{-2\beta\omega((j,i))}$ or
$K_{ij}=-e^{-2\beta\omega((j,i))}$. Thus, the boundary orientation defines four different Kasteleyn
matrices which are labelled by the chosen signs $K^{++}$, $K^{+-}$, $K^{-+}$ and $K^{--}$. The
Pfaffian of each of these Kasteleyn matrices contains summands which correspond to non trivial
domain wall loops which can not be related to spin configurations on the Ising model. These summands
cancel if summed over all boundary orientations. Only the $(e,e)$ components participate which leads
to: $2Z=Pf(K^{++})+Pf(K^{+-})+Pf(K^{-+})+Pf(K^{--})$. 

\begin{figure}
    \centering
    \includegraphics[width=0.4\linewidth]{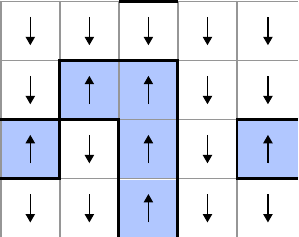}
    \caption{Relative domain walls of a spin configuration on $G^{\star}$. To describe a spin
    configuration uniquely we have to define relative domain walls and additionally the spin value
    on a single face.}
    \label{fig:relative_domain_walls}
\end{figure}
\begin{figure}
    \centering
    \includegraphics[width=0.5\linewidth]{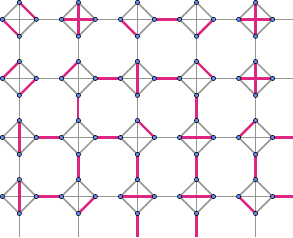}
    \caption{Mapping of domain walls to a dimer cover of $\tilde{G}^{\star}$ with Kasteleyn city
    decoration on original vertices in $G^{\star}$. Observe that a dimer configuration is unique at
    cities with inter city connections determined by the respective domain wall. Vertices not
    participating in a domain wall can be covered in three different ways by dimers. These three
    possibilities are shown on the three cities on the top left corner. Thus, the map is not one to one
    between domain walls and dimer coverings. Furthermore, the map is invariant under global spin flip
    and thus not injective.}
    \label{fig:dimer_cover}
\end{figure}

\section{Replica-exchange Wang-Landau algorithm}%
\label{WL_appendix}

The replica-exchange Wang-Landau algorithm~\cite{vogel_generic_2013, vogel_practical_2018}, is a
Monte Carlo algorithm to estimate the density of states $g(E_{i})$ over an energy spectrum divided
into bins $E_{i}$ of a classical spin system described by the Hamiltonian $H$. We summarize the
algorithm in Alg.~\ref{alg:wl}.

The algorithm relies on the fact that a random walk over the energy spectrum $S=\{E_{i}\}$ with
probability proportional to $\frac{1}{g(E_{i})}$ will produce flat histograms with support on the
energy spectrum of $H$. Access to the density of states enables the calculation of the partition
function for arbitrary temperatures, including the zero-temperature limit, in contrast to the
finite-temperature restriction of the FKT algorithm.
We start by explaining the general procedure of the Wang-Landau (WL) algorithm and continue by
explaining the chosen replica exchange (RE) parallelization scheme.  The WL algorithm starts by
initializing an array $\log(g(E_{i}))=0~\forall~E_{i}\in S$, which holds the density of states
estimate after completing the algorithm and a histogram $h(E_{i})=0~\forall~E_{i}\in S$. In
addition, one initializes an update factor $f=e$. Subsequently, random spin flips are performed. The
move is accepted with transition probability $p(E_{1}\rightarrow
E_{2})=\min\{\frac{g(E_{1})}{g(E_{2})}, 1\}$, with energy of the spin configuration $E_{1}$ before
the spin flip and $E_{2}$ after the flip. For each update step, a potentially new configuration with
energy $E$ is realized. The histogram and density arrays are updated accordingly: $h(E) \mapsto h(E)
+ 1$ and $\log(g(E)) \mapsto \log(g(E)) + \log(f)$. The random walk continues until the histogram
satisfies the flatness condition determined by a flatness parameter $\alpha$: $\min(h(E_{i})) \geq
\alpha \cdot \text{mean}(h(E_{i}))$. If the histogram is flat under this condition, the update
parameter is modified by $f\mapsto \sqrt{f}$ and the histogram is reset $h(E_{i})=0~\forall~E_{i}\in
S$. The random walk continues as long as the update factor $f$ is not below a threshold determined
by a run parameter $\beta$: $f \geq e^{\beta}$. The parameter $\beta$ controls the precision of the
simulation and thus the total number of random spin flips. The resulting $\log(g(E_{i}))$ contains
relative degeneracy factors that must be rescaled to match physical degrees of freedom. We are using
the total number of spin configurations to rescale the degeneracy factors by $\sum_{E_{i} \in
S}g(E_{i})\overset{!}{=}2^{\text{\#spins}}$. The algorithm requires knowledge on the energy
spectrum, which is estimated by performing random walks over the energy space with the WL transition
probability over a fixed number of steps within our investigation.

\begin{algorithm}
\caption{Wang-Landau Algorithm}
\KwIn{$\{E_i\}$: Energy spectrum, $H$: Hamiltonian, $\alpha$: Flatness condition, $\beta$: 
Minimal update factor, $N$: Number of MC steps per iteration}
\KwOut{Density of states $\log g(E_{i})$}
\SetKwInOut{Input}{Input}
\SetKwInOut{Output}{Output}
\BlankLine
Initialize density estimates $\log g(E_i) = 0~\forall~E_i$ \\
Initialize histogram $h(E_i) = 0~\forall~E_{i}$\\
Initialize update factor $f = e$ \\
Initialize spin system S\\
Calculate energy E of S\\
\While{$f \geq e^{\beta}$}{
    \For{N times}{
        Random spin flip
        Calculate energy after flip $\tilde{E}$\\
        Accept spin flip with transition probability $p=\min\left\{\frac{g(E)}{g(\tilde{E})}, 1\right\}$\\
        Store new E\\
        $h(E) += 1$\\
        $\log(g(E)) += \log(f)$\\
    }
    \If{$\min(h) \geq \alpha \cdot \text{mean}(h)$}{
        Update $f \mathrel{=}{\sqrt{f}}$\\
        Reset $h(E_{i})=0~\forall~E_{i}$
    }
}
\label{alg:wl}
\end{algorithm}

\FloatBarrier
\edit{
\section{Examples of confidence interval ratios}\label{app:CIratio}

Examples of ratios of confidence interval widths for estimation of logical error rates through counting failures of a maximum partition function (``counting method) and through the decoding probability (``ratio method) are shown in Fig.~\ref{fig:CIratio} (toric code under bitflip noise) and Fig.~\ref{fig:CIratio_cc} (color code under bitflip and depolarizing noise). 
The confidence interval widths for the counting method are displayed both using bootstrapping with 1000 resamples and using Jeffreys interval. We find that the confidence interval widths between the counting method and the ratio method differ by a scale factor $r(p,d)$ that is independent of the total number of samples, with sample-dependent fluctuations that vanish with increased sample size. The values of $r(p,d)$ shown in Figs.~\ref{fig:ratio_p_dependency} and \ref{fig:color_code_CI_ratios} are obtained by taking the mean over the values obtained from subsets of the sample that include more than $0.1$ of the full sample size. The error bars on $r(p,d)$ are obtained by the corresponding standard deviation.

The independence of $r(p,d)$ from the sample size is the basis of the ansatz used in the interpolation of the confidence intervals of the ratio method in Figs.~\ref{fig:confidence_intervals} and \ref{fig:color_code_CI_widths}, and to estimate the relative required sample sizes $R$ to reach a given target confidence interval $\Delta$, as discussed in Section~\ref{FKT_results_section}.

\begin{figure}
    \centering
    \includegraphics[width=0.5\linewidth]{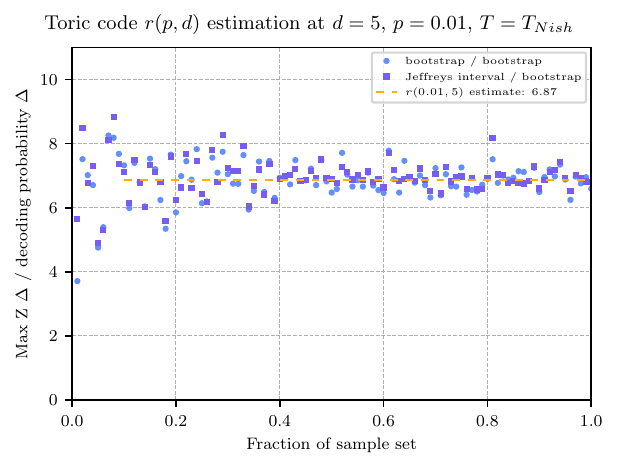}
    \caption{
    \edit{
    Ratios of confidence interval widths $\Delta$ for estimates of optimal logical error rates by counting method and by ratio method, for the toric code under bitflip noise, using subsamples containing different fractions of the full sample size of just below two million physical error configurations.
    The values of $\Delta$ for the counting method are generated using both Jeffreys interval and bootstrapping, while the values of $\Delta$ for the ratio method are generated using bootstrapping only. The labels ``bootstrap/bootstrap'' and ``Jeffreys interval/bootstrap'' indicate what method was used for counting method/ratio method to generate a given data point. 
    We derive an estimate of the scale factor $r(0.01,5)$ by taking the mean of the data points (yellow line). When computing the mean we exclude data points at sample set fractions below $0.1$, since these contain very few logical failures and lead to a large uncertainty in the counting method. 
    Estimates of $r(p,d)$ for varying $p,d$, derived in the same manner, are displayed in Fig.~\ref{fig:ratio_p_dependency}. }}
    \label{fig:CIratio}
\end{figure}

\begin{figure}
    \centering
    \text{\scriptsize Color code $r(p,d)$ estimation at $T=T_{Nish}$}
    \vfill
    \begin{subfigure}{0.49\textwidth}    
        \includegraphics[width=\linewidth]{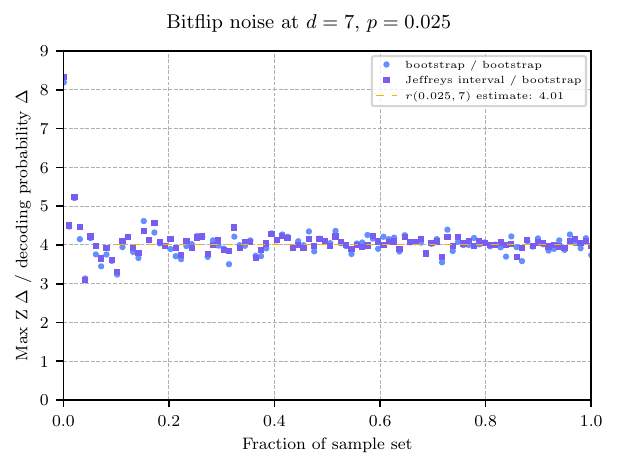}%
    \end{subfigure}
    \hfill
    \begin{subfigure}{0.49\textwidth}      
        \includegraphics[width=\linewidth]{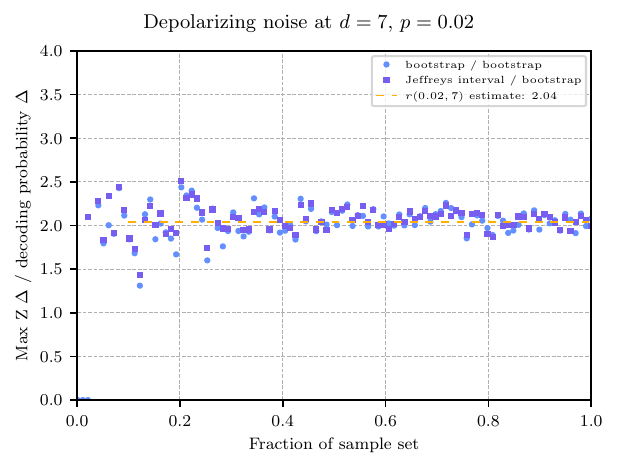}
    \end{subfigure}
    \caption{
    \edit{
        Ratios of confidence interval widths $\Delta$ for estimates of optimal logical error rates by counting method and by ratio method, for the color code under bitflip noise (left) and depolarizing noise (right), using subsamples containing different fractions of the full sample sizes, which are $5\times 10^5$ and $8\times 10^5$, respectively. Same method is used as in Fig.~\ref{fig:CIratio} to estimate the corresponding values of $r(p,d)$.
        Estimates of $r(p,d)$ for varying $p,d$, derived in the same manner, are displayed in Fig.~\ref{fig:color_code_CI_ratios}.
    }
    }
    \label{fig:CIratio_cc}
\end{figure}
}

\FloatBarrier

\section{Minimum-Weight Perfect Matching Decoder}%
\label{app:mwpm_decoder}

A Minimum-weight perfect matching (MWPM) decoder is a maximum probability decoder, that finds a
correction operator $e(\vec{s})$ consistent with a syndrome $\vec{s}$, that is
\begin{equation}
    \label{eq:mwpm_maxprob}
    P(e(\vec{s}))\geq P(e^\prime(\vec{s})), \quad\forall e^\prime(\vec{s}).
\end{equation}
We limit our discussion to the toric code as well as the unrotated and rotated planar surface codes,
though MWPM decoding is applicable to a considerable number of quantum error correction
codes~\cite{higgot2022}. Decoding is generally performed independently for Pauli-$X$ and Pauli-$Z$
errors, by the same means. Given a Pauli-$Z$ error of the form $e\in\{Z,I\}^m$, we denote by
$\vec{e}\in\mathbb{Z}_2^m$ the binary vector with $\vec{e}_i=1$ if an error occurred at qubit $i$
and $\vec{e}_i=0$ otherwise. The error models is fully described by a  three objects.  The detector
check matrix $H\in\mathbb{Z}_2^{n\times m}$ has a row for each detector measurement and a column for
each error mechanism, with $H_{ij}=1$ if detector $i$ is flipped by error mechanism $j$ and
$H_{ij}=0$ otherwise. Each error mechanism $i$ occurs with probability $p_i$, described by
$\vec{p}\in[0,1]^m$. Furthermore, the effect of errors on the logical observables is captured in
$O\in\mathbb{Z}_2^{n_l\times m}$, with $O_{ij}=1$ if the logical observable $i$ is flipped by error
$j$ and $O_{ij}=0$ otherwise. Such a general description of an error model by $H$, $O$ and $\vec{p}$
is readily available from e.g., a stabilizer circuit simulator such as
\texttt{Stim}~\cite{gidney2021} by forward propagating Pauli errors through the circuit.  The steps
of the decoding procedure are shown in Fig.~\ref{fig:mwpm_algo}. The error model induces the
matching graph $ \mathcal{G}_M$ with incidence matrix $H$, where each stabilizer measurement is a
node and each error mechanism an edge. For codes that admit error chains with a single defect, such
as the rotated and unrotated planar surface code for error chains ending at the edge of the lattice,
an additional boundary node is introduced for each such stabilizer measurement, with all boundary
nodes connected by edges of weight $0$. The remaining edges in the graph are weighted
by~\cite{dennis2002b}
\begin{equation}
    w_i = \log\frac{1-p_i}{p_i},
\end{equation}
such that edges corresponding to more probable errors have lower weights.  Every error
$\vec{e}\in\mathbb{Z}_2^n$ produces a syndrome $\vec{s}=H\vec{e}$ which needs to be decoded into a
correction operator $\vec{g}\in\mathbb{Z}_2^m$ with $\vec{g}\equiv e(\vec{s})$ fulfilling
Eq.~\eqref{eq:mwpm_maxprob}. From the defect nodes in $\vec{s}$ and the matching graph
$\mathcal{G}_M$ the syndrome graph $\mathcal{G}_{\vec{s}}$ shown in Fig.~\ref{sfig:mwpm_sgraph} is
constructed. It is a complete graph of all the defect nodes and the shortest paths between them in
$\mathcal{G}_M$ as edges. A matching in $\mathcal{G}_{\vec{s}}$ is a set of edges in
$\mathcal{G}_{\vec{s}}$, such that no two edges are incident to the same node. A perfect matching is
a matching containing all nodes of the syndrome graph, it also is a minimum-weight perfect matching,
if it fulfills these conditions and the sum of the edge weights contained in the matching is
minimal. Any minumum-weight perfect matching in $\mathcal{G}_{\vec{s}}$ corresponds to a correction
operator $\vec{g}\in\mathbb{Z}_2^m$ that is consistent with the syndrome and furthermore fulfills
Eq.~\eqref{eq:mwpm_maxprob}. The decoding is successful, if $\vec{g}$ and the (unknown) error
$\vec{e}$ that caused syndrome $\vec{s}$ have the same effect on $O$, that is
\begin{equation}
    O(\vec{e}\oplus\vec{g}) = 0.
\end{equation}
Otherwise a logical error occurs. For the efficient computation of a minimum-weight perfect
matching, Edmonds' blossom algorithm~\cite{edmonds1965} is the common method, of which various
high-quality implementations and variants exist~\cite{kolmogorov2009, dezsHo2011, higgot2022,
wu2023, higgott2025}. Here we conduct all simulations involving minimum-weight perfect matching
decoding with the \texttt{PyMatching 2} software package, implementing the sparse blossom
algorithm~\cite{higgott2025}.

\begin{figure}[t]
    \centering
    \subcaptionbox{Surface Code.\label{sfig:mwpm_code}}
        {\includegraphics{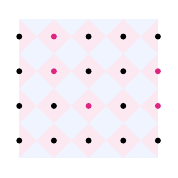}}%
    \hfill
    \subcaptionbox{Matching Graph.\label{sfig:mwpm_mgraph}}
        {\includegraphics{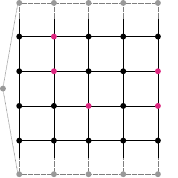}}%
    \hfill
    \subcaptionbox{Syndrome Graph.\label{sfig:mwpm_sgraph}}
        {\includegraphics{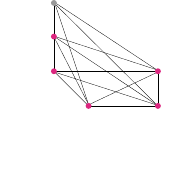}}%
    \hfill
    \subcaptionbox{Matching.\label{sfig:mwpm_match}}
        {\includegraphics{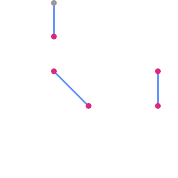}}%
    \hfill
    \subcaptionbox{Correction.\label{sfig:mwpm_corr}}
        {\includegraphics{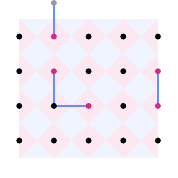}}%
    \caption{The Pauli-$X$ stabilizer measurements of the distance 5 surface
        code~(\subref{sfig:mwpm_code}), induce a matching graph~(\subref{sfig:mwpm_mgraph}) which is
        extended by boundary nodes. From the matching graph, the syndrome graph is constructed as a
        complete graph of the defect nodes (red) and the shortest paths between
        them~(\subref{sfig:mwpm_sgraph}). A minimum-weight perfect matching in the syndrome
        graph~(\subref{sfig:mwpm_match}) corresponds to a most probable correction operator
        consistent with the syndrome~(\subref{sfig:mwpm_corr}).}
    \label{fig:mwpm_algo}
\end{figure}

\edit{\section{Further results for MWPM Decoding and Ensembling under other noise models}%
\label{app:further_plots}

This Appendix contains further plots for Section~\ref{OtherNoise}. 
Fig.~\ref{fig:hook_errors_d4} shows the effect of ensembling on the rotated surface code under split hook error bitflip noise at $d=4$.
The results for $d=3$ are not shown since there are only two error locations at this distance, and no combination of errors on these two locations leads to a logical error. 
Fig.~\ref{fig:readout_noise} shows the effect of ensembling on the rotated surface code under combined bitflip and readout noise, with $d$ rounds of stabilizer measurements. The readout noise is applied to $Z$-type stabilizers only.

\begin{figure}
    \centering
    \includegraphics[width=\linewidth]{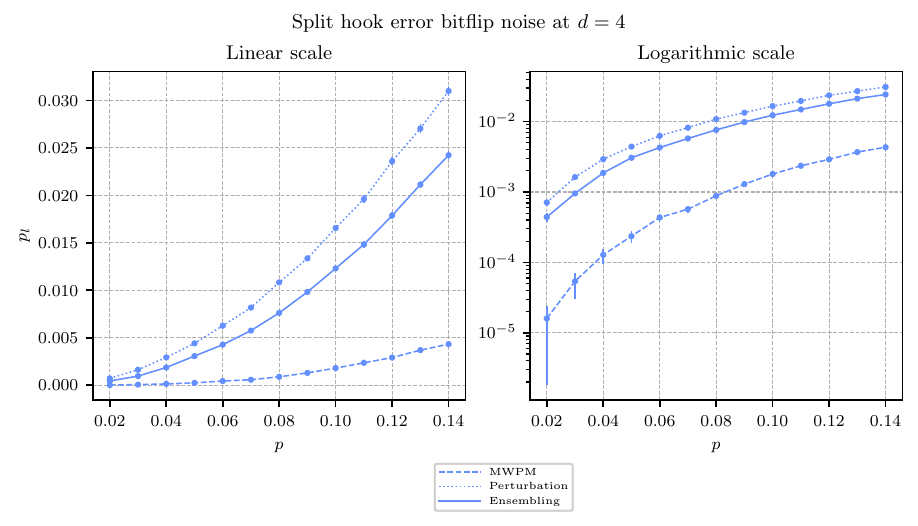}
    \caption{\edit{Left: Logical performance of the surface code under split bitflip hook errors on $X$-type stabilizers for  $d=4$. Right: The same logical performance in log-log-scale, allowing the low-$p$ behavior to be seen. The ensembling uses 50 random perturbations. The results labeled ``Perturbation'' were generated using one single perturbation. MWPM significantly outperforms both ensembled MWPM and perturbed MPWM, indicating that the bias in \texttt{PyMatching} aligns with the split hook errors at $d=4$.}}
    \label{fig:hook_errors_d4}
\end{figure}

\begin{figure}
    \centering
    \includegraphics[width=\linewidth]{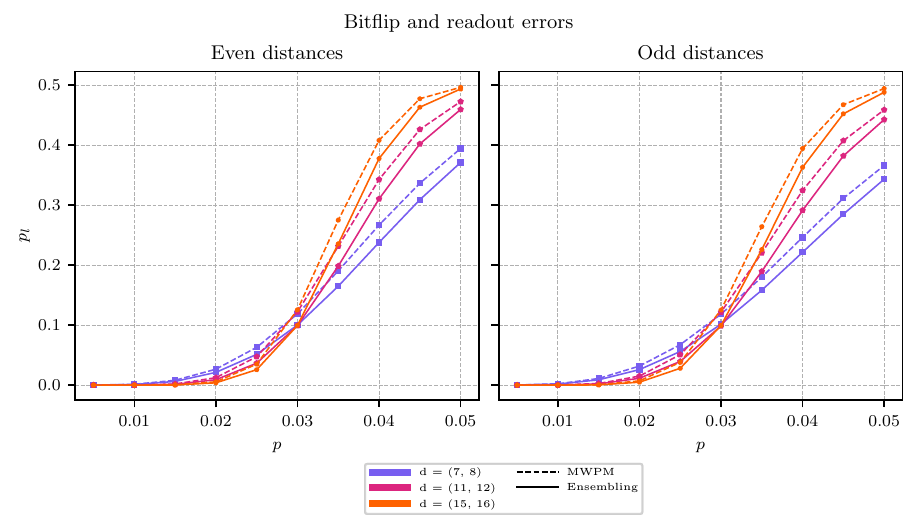}
    \caption{\edit{
    Logical performance of the surface code under combined bitflip and readout noise for even (left) and odd (right) code distances. The ensembling uses 50 random perturbations on the decoding graph. Both even and odd distances show modest performance improvement under ensembling. The observed threshold is compatible with the known value of $p_{th}\sim 2.9\%$ \cite{wang2003}.
    }}
    \label{fig:readout_noise}
\end{figure}

}

\section{Optimization of Ensembling Parameters}%
\label{app:opt_ensembling}

Ensembling by perturbing the matching graph of the MWPM decoder as discussed for a uniform noise
model in Section~\ref{section:modifiedMatchingBasedDecoder} and a non-uniform noise model in
Section~\ref{ssec:nonuniform_mwpm}, requires setting the standard deviation for the edge weight
perturbations $\sigma_\xi$ and $\sigma_\Xi$. We optimized both values using a simple coarse grained
global optimization scheme~\cite{jones1993} within the interval $[0, 0.5]$ for both standard
deviations and furthermore tested multiple values within the neighborhood of the determined
approximate optimum. Due to the considerable computational cost of this optimization, we limited the
search to the toric code with distances $d=8$ and $d=16$, error rates of $p\in\{0.05, 0.07, 0.1,
0.11\}$, standard deviations for the non-uniform noise model of $\sigma_p\in\{0.005, 0.01, 0.02, 0.03,
0.04, 0.05, 0.1, 0.15\}$, and 10,000 sampled error configurations for each optimization step. For
the uniform noise model, we consistently find the best results for $\sigma_\xi\in[10^{-6},10^{-2}]$,
for larger values the observed logical error rate increases monotonously, considerably smaller
values lead to slight increases. The optimization results for the uniform noise model are shown in
Fig.~\ref{fig:mwpm_opt_pert_uniform}.  With a truncated Gaussian noise model, across the tested
error rate standard deviations, an approximate optimum for all tested configurations can be observed
at $\sigma_\Xi\approx 0.025$, which we set for all ensembling simulations with the MWPM decoder and
non-uniform noise-models. The optimization results for the toric code with $d=8$ and $d=16$ are
shown in Fig.~\ref{fig:mwpm_opt_pert_nonuniform_8} and Fig.~\ref{fig:mwpm_opt_pert_nonuniform_16}.

\begin{figure}
    \centering
    \includegraphics{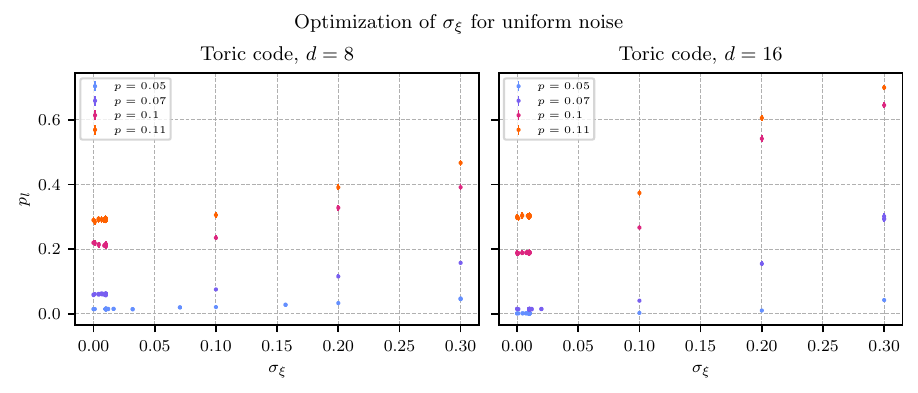}
    \caption{Optimization results for $\sigma_\xi$ under a uniform noise model and the toric code
    with distance $d=8$ (left) and $d=16$ (right). Across four different error rates $p$ we find the
    best results for $\sigma_\xi\in[10^{-6},10^{-2}]$.}
    \label{fig:mwpm_opt_pert_uniform}
\end{figure}

\begin{figure}
    \centering
    \includegraphics{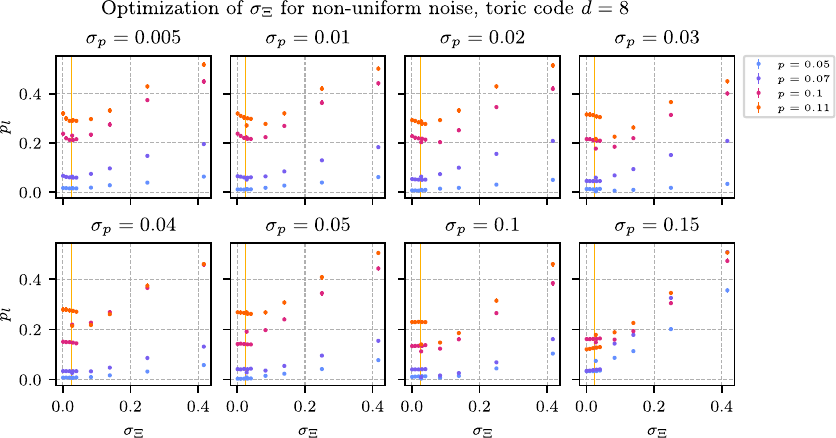}
    \caption{Optimization results for $\sigma_\Xi$ under a non-uniform noise model with various
    noise standard deviations $\sigma_p$ and the toric code with distance $d=8$. The approximate
    optimum of $\sigma_\Xi=0.025$ is marked with the vertical orange line.}
    \label{fig:mwpm_opt_pert_nonuniform_8}
\end{figure}

\begin{figure}
    \centering
    \includegraphics{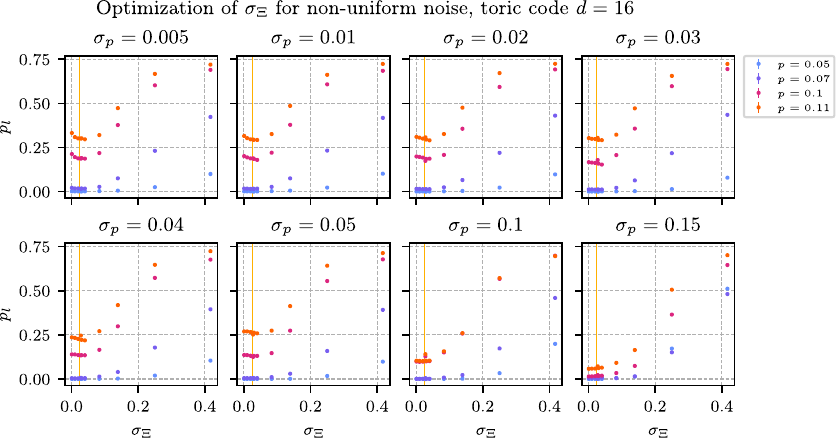}
    \caption{Optimization results for $\sigma_\Xi$ under a non-uniform noise model with various
    noise standard deviations $\sigma_p$ and the toric code with distance $d=16$. The approximate
    optimum of $\sigma_\Xi=0.025$ is marked with the vertical orange line.}
    \label{fig:mwpm_opt_pert_nonuniform_16}
\end{figure}

\clearpage

\bibliographystyle{quantum.bst}
\bibliography{bibliography}
\end{document}